\pdfoutput=1
\documentclass[conference]{IEEEtran}
\IEEEoverridecommandlockouts
\usepackage{cite}
\usepackage{amsmath,amssymb,amsfonts}
\usepackage{algorithmic}
\usepackage{graphicx}
\usepackage{textcomp}
\usepackage{xcolor}
\def\BibTeX{{\rm B\kern-.05em{\sc i\kern-.025em b}\kern-.08em
    T\kern-.1667em\lower.7ex\hbox{E}\kern-.125emX}}

\usepackage{hyperref}
\hypersetup{colorlinks=true, unicode=true, linkcolor=[rgb]{0.10,0.05,0.67}, citecolor=[rgb]{0.10,0.05,0.67}, filecolor=[rgb]{0.10,0.05,0.67}, urlcolor=[rgb]{0.10,0.05,0.67}}

\hypersetup{
  pdftitle={Maximally Extendable Product Codes are Good Coboundary Expanders},
  pdfauthor={Gleb Kalachev and Pavel Panteleev},
  pdfkeywords={high-dimensional expanders, tensor product codes, quantum LTC}
}

\usepackage{amsthm}
\usepackage{thmtools}

\usepackage{tikz}
\usetikzlibrary{cd}
\usetikzlibrary{decorations.pathreplacing}
\usetikzlibrary{patterns}
\usepackage{ifthen}

\tikzset{vertex/.style={ draw , circle , fill , inner sep=0em , minimum size=0.3em}}
\tikzset{empty/.style={inner sep=0em, outer sep=0em, minimum size=0em}}

\newcommand{\Iff}{\textbf{if\textcompwordmark f} }

\newcommand{\rbr}[1]{\left(#1\right)}

\newcommand{\abr}[1]{\left\langle#1\right\rangle}
\newcommand{\floorbr}[1]{\left\lfloor #1\right\rfloor}

\newcommand{\fbr}[1]{\left\{#1\right\}}

\newcommand{\bigabsbr}[1]{\bigl|#1\bigr|}
\newcommand{\bigrbr}[1]{\bigl(#1\bigr)}

\let\le\leqslant
\let\ge\geqslant
\let\phi\varphi

\newcommand{\eps}{\varepsilon}

\theoremstyle{plain}
\newtheorem{theorem}{Theorem}
\newtheorem{proposition}[theorem]{Proposition}
\newtheorem{lemma}[theorem]{Lemma}
\newtheorem{corollary}[theorem]{Corollary}
\newtheorem{claim}{Claim}[theorem]

\newtheorem{exttheorem}{Theorem}

\newtheorem{extlemma}[exttheorem]{Lemma}

\theoremstyle{definition}
\newtheorem{definition}[theorem]{Definition}

\theoremstyle{remark}
\newtheorem{remark}[theorem]{Remark}

\newcommand{\NN}{\mathbb{N}}

\newcommand{\bC}{\mathbf{C}}

\newcommand{\vp}{\mathbf{p}}
\newcommand{\va}{\mathbf{a}}

\newcommand{\abs}[1]{\left\lvert #1 \right\rvert}
\newcommand{\norm}[1]{\left\lVert#1\right\rVert}

\newcommand{\cA}{\mathcal{A}}
\newcommand{\cB}{\mathcal{B}}
\newcommand{\cC}{\mathcal{C}}
\newcommand{\cF}{\mathcal{F}}
\newcommand{\cG}{\mathcal{G}}
\newcommand{\cL}{\mathcal{L}}

\newcommand{\cay}{\mathrm{Cay}}

\newcommand{\Z}{\mathbb{Z}}
\newcommand{\F}{\mathbb{F}}

\DeclareMathOperator{\rk}{rk}
\DeclareMathOperator{\im}{im}
\DeclareMathOperator{\supp}{supp}

\newcommand{\Gr}{{\sf Gr}}
\newcommand{\PP}{\mathsf{P}}

\begin{document}

\title{Maximally Extendable Product Codes are\\ Good Coboundary Expanders\\
\thanks{This work was supported by the Ministry of Economic Development of the Russian Federation in accordance with the subsidy agreement (agreement identifier
000000C313925P4H0002; grant No 139-15-2025-012).}
}

\author{\IEEEauthorblockN{Gleb Kalachev}
\IEEEauthorblockA{\textit{Faculty of Mechanics and Mathematics} \\
\textit{Moscow State University}\\
Moscow, Russia\\
gleb.kalachev@yandex.ru
}
\and
\IEEEauthorblockN{Pavel Panteleev}
\IEEEauthorblockA{\textit{Faculty of Mechanics and Mathematics} \\
\textit{Moscow State University}\\
Moscow, Russia\\
panpavel@yandex.ru
}
}

\IEEEaftertitletext{%
\vspace{-1.2\baselineskip}%
\noindent\emph{This is the author’s version of the paper accepted for publication in the Proceedings of the 66th IEEE Symposium on Foundations of Computer Science (FOCS 2025). The final version will appear at IEEE Xplore.}\par\medskip
}

\maketitle

\begin{abstract}
We investigate the coboundary expansion property of tensor product codes, known as product expansion, which plays an important role in recent constructions of good quantum LDPC codes and classical locally testable codes. Prior research has shown that this property is equivalent to agreement testability and robust testability for products of two codes with linear distance. However, for products of more than two codes, product expansion is a strictly stronger property. In this paper, we prove that a collection of an arbitrary number of random codes over a sufficiently large field has good product expansion. We believe that, in the case of four codes, the same ideas can be used to construct good quantum locally testable codes, in a way similar to the current constructions that use only products of two codes.
\end{abstract}


\section{Introduction}

The study of tensor product codes and their coboundary expansion properties has gained significant attention in recent years due to their key role in the discovery of classical locally testable codes (LTCs)~\cite{Dinur:stoc2022, Panteleev&Kalachev:stoc2022} and good quantum low-density parity-check (qLDPC) codes~\cite{Panteleev&Kalachev:stoc2022}. Already in~2014, it was noticed that simplicial complexes with good coboundary expansion over $\F_2$, usually referred to as \emph{high-dimensional expanders} (HDXs), could potentially be used to construct LTCs with large soundness~\cite{Kaufman:2014a} and qLDPC codes with large distances~\cite{Kaufman:2014, Evra:2020}. 

Inspired by these works and the breakthrough paper on fiber bundle codes~\cite{Hastings:2021:fiber}, a~construction called \emph{expander lifted product codes} was proposed in~\cite{Panteleev&Kalachev:stoc2022}, which combines the ideas of high-dimensional \emph{expansion} with \emph{lifted product codes}~\cite{Panteleev&Kalachev:2019,Panteleev&Kalachev:2021} to simultaneously produce (asymptotically) good qLDPC codes and good LTCs with large soundness. 
Lifted product (LP) codes, initially introduced in QEC'2019~\cite{Panteleev&Kalachev:2019} as a~natural generalization\footnote{In~\cite{Panteleev&Kalachev:2019} they were used in a~slightly less general form and termed \emph{generalized hypergraph product} codes.} of hypergraph product codes~\cite{Tillich&Zemor:2014} and other constructions~\cite{qldpc,Hagiwara:2007,Haah:2011,Kovalev:2013}, were conceptualized and reformulated in~\cite{Panteleev&Kalachev:2021} as a~general framework to construct qLDPC codes using tensor products $\cA \otimes_R \cB$ of chain complexes over a~\emph{ring} $R$ such as a~group algebra\footnote{In fact, this idea works with arbitrary noncommutative rings $R$ (see the footnote in \cite[p.~3]{Panteleev&Kalachev:2021}).} $\F_2[G]$. To obtain binary codes, one views vectors over $R$ as $| G |$ times larger vectors over~$\mathbb{F}_2$, representing each code symbol $\sum_{g \in G} a_g g \in R$ by the block $(a_g)_{g \in G}$ of $| G |$ bits.

But where can we find chain complexes $\cA$ and $\cB$ over $R$? Fortunately, classical LDPC codes constructed using $G$-lifts of graphs for a~group $G$ (e.g., see~\cite[p.~3]{Panteleev&Kalachev:2021}), can also be viewed as codes over $R = \F_2[G]$ instead of~$\F_2$, and thus can be represented by the $2$-term chain complex $R^n\xrightarrow{H} R^m$, where $H\in R^{m\times n}$ is a~parity-check matrix over~$R$.
Important examples of such codes are \emph{Z\'{e}mor codes}~\cite{Zemor:2001}, a~variation of Sipser-Spielman codes~\cite{Sipser:1996} that are usually defined\footnote{For an introduction to Z\'{e}mor codes and the associated Z\'{e}mor decoder, please refer to Spielman’s lectures~\cite{Spielman2009SpectralGraphTheory} or Section~6.2.1 of the lectures by Dwork \& Harsha~\cite{dworkCS369EExpandersComputer}.} as Tanner codes~\cite{Tanner:1981} based on double covers of Ramanujan Cayley graphs~\cite{Margulis:1988, Lubotzky:1988}. These codes were further generalized to lifts of general bipartite graphs in~\cite{Panteleev&Kalachev:2021} and suggested as component codes in the lifted product construction, which was eventually utilized in~\cite{Panteleev&Kalachev:stoc2022} to obtain good LTCs and qLDPC codes.

It was already proved in~\cite[p.~8]{Panteleev&Kalachev:2021} that by combining two \emph{constant-rate} codes using the lifted product construction one can get a~\emph{constant-rate} code, while the question of whether we can also achieve linear distance remained widely open~\cite[p.~14]{Panteleev&Kalachev:2021}. The expander LP codes~\cite{Panteleev&Kalachev:stoc2022} solve this open problem by combining two constant-rate Z\'{e}mor codes in a~way that ensures that the corresponding 3-term chain complex has good high-dimensional expansion in both directions. Then, adopting the methods from the theory of HDXs~\cite{Kaufman:2014a, Kaufman:2014, Evra:2020}, it is possible to show that this implies good \emph{distance} of the qLDPC code and simultaneously good \emph{soundness} of the LTCs from this 3-term complex. However, for this construction to work, it is \emph{required} that the pair of local codes in the Z\'{e}mor construction has a~very specific property called \emph{product expansion} that can be expressed topologically as the~coboundary expansion of the sheaf on the~complete bipartite graph naturally associated with the product code~\cite{kalachevTwosidedRobustlyTestable2023}. The proof in~\cite{Panteleev&Kalachev:stoc2022} shows that this local expansion lifts to the~global one, which gives the desired properties. A~similar approach to lift small tensor product codes to larger ones using HDXs was used in the~independent work on good LTCs~\cite{Dinur:stoc2022} (see also~\cite{LiftingSmallLocally2019}).   
Later it was shown that expander LP codes can also be converted into other examples of good qLDPC codes~\cite{Leverrier:focs2022, Dinur:decoders} that have more natural geometric interpretations, and all these codes have linear time decoding algorithms~\cite{Gu:stoc2023:qpdpc-decoder, Leverrier:qldpcdecoder:2023, Dinur:decoders, guSingleShotDecodingGood2024} as originally conjectured in~\cite{Panteleev&Kalachev:stoc2022}. 

The expansion properties of product codes in more than two dimensions are very important in the context of quantum LTCs~\cite{Aharonov:2015}. One possible way to address the \emph{qLTC conjecture}~\cite{eldarLocalHamiltoniansWhose2017,crossQuantumLocallyTestable2024, dinurqLTC:2024}, positing the existence of good quantum locally testable codes (qLTCs), is to consider a~$4$-dimensional analog of the expander LP codes. One natural way is to use the tensor product of \emph{four} generalized Z\'{e}mor codes instead of \emph{two}.  In the~recent result~\cite{dinurqLTC:2024}, it is shown that using this idea one can get \emph{almost} good qLTCs, provided that both the collection of the four local codes and their duals have good product expansion. There are also two more recent works relying on the product expansion of more than two codes~\cite{golowichQuantumLDPCCodes2024, nguyenQuantumFaultTolerance2024}.

In this paper, we investigate the product expansion property of an~\emph{arbitrary} number of codes. In the case of two codes, product expansion is equivalent to agreement testability and robust testability for products of two codes with linear distance~\cite{kalachevTwosidedRobustlyTestable2023}. However, it is important to note that in the case of more than two codes, it is a~strictly \emph{stronger} property than both the agreement and the robust testability~\cite{Kalachev:example:2023}. The fact that the product expansion of two random codes is good was already proved in~\cite{Dinur:decoders, kalachevTwosidedRobustlyTestable2023}. Our main result shows that a~collection of an~arbitrary number of random codes over a~sufficiently large field has good product expansion with high probability.
Specifically, we prove that for any fixed rates, there exists a~positive constant $\rho$ such that a~randomly chosen collection of codes over a~sufficiently large field is $\rho$-product-expanding, where the field size grows with the code length.

Our work builds on several new key insights and techniques. First, we establish that \emph{locally testable codes} (LTCs) can be used to construct product-expanding collections of codes. We show that if a~collection of codes consists of LTCs, then it is product-expanding. However, since the duals of LTCs have small minimum distance, we cannot directly use LTCs to construct product-expanding collections of both the codes and their duals. To overcome this limitation, we use the concept of \emph{maximally extendable product codes}~\cite{panteleevMaximallyExtendableSheaf2024}, which are product codes that inherit the extendability properties of all other product codes with the same parameters. We provide a~reformulation of the product expansion property in terms of extendable subsets of the product of dual codes, which simplifies the analysis of product expansion in higher dimensions.
We prove that random codes over sufficiently large fields are maximally extendable with high probability, allowing us to construct product-expanding collections of both the codes and their duals.

In addition to our main result, we provide several useful auxiliary results. First, we show that LTCs of \emph{arbitrary} length and with the rate approaching $1$ exist, which allows us to perform rate adaptation in our constructions. Second, we prove that if we take subspaces of a~product-expanding collection of codes, the resulting collection of subcodes remains product-expanding, albeit with a~worse constant.

\section{Definitions and Main Results}
In this section, we define a~property of a~collection of codes called \emph{product expansion}. Before we proceed, let us first briefly recall some standard definitions that can also be found in~\cite{kalachevTwosidedRobustlyTestable2023}. 
We consider only finite fields $\F_q$ of characteristic~$2$. 
Given a~set $S$ and a~field~$\F_q$, denote by $\F_q^S$ the vector space of all formal $\F_q$-linear combinations $v = \sum_{s\in S} v(s)\cdot s$, whose elements are also represented as vectors $(v(s))_{s\in S}$ over $\F_q$ indexed by the set~$S$. According to our definition, $S \subseteq \F_q^S$ is a~basis of the vector space $\F_q^S$, and $\F_q^S \subseteq \F_q^T$ whenever $S\subseteq T$. If $v\in\F_q^B$ and $A\subseteq B$, then by $v|_A$ we denote the \emph{restriction} of $v$ to $A$. A~matrix is called \emph{$\Delta$-limited} if for all its rows and columns the number of non-zero elements is at most $\Delta$.  Also, denote by $I_n$ the identity $n\times n$ matrix. 

Given linear codes $\cC_1,\dots,\cC_D \subseteq \F_q^n$ we can define the (\emph{tensor}) \emph{product code} 
\[
\cC_1\otimes\dots\otimes\cC_D := \{c\in \F_q^{[n]^D} \mid \forall i\in [D]\ \forall \ell\in \cL_i\colon c|_\ell\in  \cC_i\},
\]
where $\F_q^{[n]^D}$ is the set of functions $c\colon [n]^D \to \F_q$  and $\cL_i = \cL_i(n, D)$ is the set of lines parallel to the $i$-th axis in the $D$-dimensional grid $[n]^D$, i.e., 
\[
\cL_i := \fbr{ A_1\times\dots\times A_D \subseteq [n]^D \mid  A_i\! =\! [n], \abs{A_j}\! =\! 1, \forall j\ne i}.
\]

It is convenient to introduce a~notation for the dual code of a~product code~\cite{Wolf:1965, Chien:1973}. For linear codes $\cC_1,\cC_2\subseteq \F_q^{n}$ we denote by $\cC_1\boxplus \cC_2$ the code 
\[(\cC_1^\bot\otimes \cC_2^\bot)^\bot=\cC_1\otimes \F_q^{n}+\F_q^{n}\otimes \cC_2 \subseteq \F_q^{n\times n}.\]
Given a~collection $\cC = (\cC_i)_{i\in [D]}$ of linear codes over~$\F_q$, we can also define the codes 
\begin{align*}
\cC^{(i)} :=&\ \underbrace{\F_q^{n} \otimes\dots \otimes \F_q^{n}}_{(i-1) \text{ times}} \otimes \ \cC_i \otimes \underbrace{\F_q^{n}\otimes \dots\otimes \F_q^{n}}_{(D-i)\text{ times}} \\
=&\ \{ c\in \F_q^{[n]^D} \mid \forall \ell\in \cL_i\colon c|_\ell\in  \cC_i \}.    
\end{align*}
It is clear that $\cC_1\otimes \dots \otimes \cC_D = \cC^{(1)}\cap \dots \cap \cC^{(D)}$ and 
\[\cC_1\boxplus \dots \boxplus \cC_D = \cC^{(1)}+ \dots + \cC^{(D)}.\] Note that every code $\cC^{(i)}$ is the direct sum of $|\cL_i| = n^{D-1}$ copies of the code $\cC_i$. For $x\in \F_q^{[n]^D}$ we denote by $|x|_i$ and $\norm{x}_i$ respectively the number and the fraction of the lines ${\ell\in\cL_i}$ such that $x|_\ell \ne 0$. Clearly, by definition, we have $\norm{x}_i = \frac{1}{n^{D-1}}|x|_i$. By $|x|$ and $\norm{x}$ we denote respectively the \emph{Hamming weight} (i.e., the number of non-zero entries) and the \emph{normalized Hamming weight} (i.e., the fraction of non-zero entries) of $x$. Now we are ready to give our main definition. 

\begin{definition}[product expansion]
Given a~collection $\cC = (\cC_i)_{i\in [D]}$ of linear codes $\cC_i\subseteq \F_q^{n}$, we say that $\cC$ is \emph{$\rho$-product-expanding} if every codeword $c\in \cC_1\boxplus \dots \boxplus \cC_D$ can be represented as a~sum $c = \sum_{i\in[D]} a_i$ where $a_i\in \cC^{(i)}$ for all $i\in [D]$, and the following inequality holds:
\begin{equation}\label{eq:prod-exp}
\rho\sum_{i\in [D]} \norm{a_i}_i \le \norm{c} .    
\end{equation}
\end{definition}

It is not hard to check that (\ref{eq:prod-exp}) can also be expressed as
\begin{equation}\label{eq:prod-exp2}
\rho \le  \frac{\abs{c}}{n\sum_{i\in [D]} \abs{a_i}_i}.    
\end{equation}
We can also define the \emph{product expansion factor} $\rho(\cC)$ for the collection $\cC$ as the maximal value of $\rho$ such that $\cC$ is $\rho$-product-expanding. In \cite{kalachevTwosidedRobustlyTestable2023} (see also~\cite{Dinur:decoders}), it is shown that a~random pair of codes is $\rho$-product-expanding for some positive $\rho$ depending only on the rates of these codes.    

A~collection $\cC = (\cC_i)_{i\in [D]}$ is called \emph{degenerate} if $\cC_i=\F_q^{n}$ for some $i\in [D]$.  In what follows, we will usually consider non-degenerate collections of codes, unless otherwise stated.
In \cite[Lemma 11]{kalachevTwosidedRobustlyTestable2023} it was shown that if a~non-degenerate collection $\cC=(\cC_i)_{i\in [D]}$ of codes $\cC_i\subsetneq \F_q^{n}$ is $\rho$-product-expanding, then each subcollection $\cC_I=(\cC_i)_{i\in I}$, $I\subseteq [D]$, is also $\rho$-product-expanding.

Denote by $\Gr_q(n,k)$ the~\emph{Grassmannian}, i.e., the set of all $k$-dimensional linear subspaces of $\F_q^n$.
In the next section, we will prove the following theorem. 

\begin{restatable}{theorem}{ThMain}\label{th:rand-expanding}
     For every tuple $(R_1,\dots,R_D)\in (0,1)^D$ there exists $\rho>0$ such that for each $n\in \NN$ a~tuple of codes $(\cC_1,\dots,\cC_D) \in \Gr_{2^t}(n,k_1)\times \cdots\times\Gr_{2^t}(n,k_D)$ picked uniformly at random, where $k_i \le n R_i$, $i\in [D]$, is $\rho$-product-expanding with probability at least $1- n^D 2^{n^D - t + 1}$, which tends to $1$ as $t\to\infty$.
\end{restatable}
Using the union bound for a~collection of codes and their duals, we get the following result\footnote{In the special case $D=1$, we use here the fact that MDS codes with the required parameters exist for all $2^{t} > n$}.
\begin{corollary}
 For any intervals $I_1,\dots,I_D\subseteq (0,1)$, there exists $\rho>0$ such that for all sufficiently large ${n\in \NN}$ there exist codes $\cC_1,\dots,\cC_D\subseteq \F_{2^{t_n}}^n$, where ${t_n=(n+1)^D}$, such that $\frac{1}{n}\dim \cC_i\in I_i$ and  
 \[
 \rho(\cC_1,\dots,\cC_D)\ge\rho,\qquad \rho(\cC_1^\bot,\dots,\cC_D^\bot)\ge\rho.
 \]
\end{corollary}

In \cite[Appendix B]{kalachevTwosidedRobustlyTestable2023} it is also shown that the product expansion  naturally generalizes the coboundary expansion~\cite{Linial:2006,Gromov:2010} of the $(D-1)$-dimensional clique complex%
    \footnote{The \emph{clique complex} $X(\cG)$ of a~graph $\cG$ is the simplicial complex with the set of vertices $V(\cG)$ where $S\in  X(\cG)$ \Iff the set $S$ gives a~\emph{clique} in $\cG$, i.e., the vertices from $S$ are pairwise connected by edges in~$\cG$.} 
$X=X(K_{n,\dots,n})$ of the complete $D$-partite graph $K_{n,\dots,n}$. 
We represent this complex as a~poset 
$X = ([n]\cup\fbr{*})^D$ with the order relation
\[
(\sigma_1,\dots,\sigma_D) \!\le\!  (\tau_1,\dots,\tau_D)  \Leftrightarrow \forall i\in[D],\ \sigma_i = *\ \text{or}\ \sigma_i = \tau_i,
\]
and the map assigning the \emph{dimension} to each \emph{cell} $\sigma\in X$:
\[
\dim \sigma = \abs{\fbr{i\in[D]\mid \sigma_i\ne *}}-1.
\]
Let $X(i) = \{\sigma\in X \mid \dim \sigma = i\}$.
It is clear that the set of  $(D-1)$-dimensional cells $X(D-1)$ is the set $[n]^D$, while $X(D-2)$ can be  identified with the set of all (axis-parallel) lines $\bigcup_{i\in[D]}\cL_i$ in the grid $[n]^D$. 

Now given a~collection $\cC=(\cC_i)_{i\in [D]}$, we define for each $\sigma \in X$ the local code
\[
\cF_\sigma  = \{c \in \F_q^{X_\sigma} \mid c|_\ell \in \cC_i, \forall \ell \in \cL_i, \ell \subseteq X_\sigma, i\in [D] \}
\] 
on the subset of coordinates 
$X_\sigma = \{x \in [n]^D \mid \sigma \le x \}$. 
The local codes $\cF_\sigma$ and the natural restrictions $\cF_{\sigma\to\tau}\colon \cF_\sigma \to \cF_\tau$, $\sigma \le \tau$, given by $x \mapsto x|_{X_\tau}$, together define the structure~$\cF$ known in algebraic topology as a~\emph{cellular sheaf} of vector spaces on the complex $X$, where the vector spaces are the linear codes $\cF_\sigma$, $\sigma \in X$. 

Now given $\cF$, we 
consider the alternating sequence called a~\emph{cochain complex}
\[
\cdots\rightarrow\bC^{i-1} \xrightarrow{\delta^{i-1}}  \bC^i \xrightarrow{\delta^{i}}  \bC^{i+1} \rightarrow \cdots 
\]
of vector spaces $\bC^i = \bigoplus_{\sigma\in X(i)} \cF_\sigma$, $i\in\Z$, and $\F_q$-linear maps  
called \emph{coboundary maps}, defined\footnote{Since the natural restriction maps $\cF_{\sigma\to\tau}\colon x\mapsto x|_{X_\tau}$ satisfy the condition 
$\cF_{\tau\to\pi} \circ \cF_{\sigma\to\tau} = \cF_{\sigma\to\pi}$
and in our poset $X$ the number of $i$-cells between every $(i-1)$-cell and $(i+1)$-cell is either $0$ or $2$, it is straightforward to check the standard condition for cochain maps: $\delta^{i}\circ\delta^{i-1} = 0$, $i\in \Z$.} as
\[
\delta^i (c_{\sigma}\cdot\sigma) = \sum_{\substack{\sigma \le \tau\\\dim \tau = i+1}} (c_{\sigma}|_{X_\tau})\cdot\tau, 
\]
and extended by linearity, where the elements from $\bC^i$ are viewed as formal linear combinations $\sum_{\sigma\in X(i)} c_\sigma\cdot \sigma$ with coefficients $c_\sigma \in \cF_\sigma$. 

The \emph{$i$-th coboundary expansion $\eta^i(X;\cF)$ of $X$ with coefficients in $\cF$} is defined as 
\begin{equation}\label{eq:cobound}
	\eta^i(X;\cF) = \min_{c\in \bC^i\setminus \im \delta^{i-1}} \frac{\abs{\delta^i c}_X}{\min_{b\in \im \delta^{i-1}}\abs{c + b}_X},
\end{equation}
where $\abs{c}_X = \abs{\{\sigma\in X(i)\mid c_\sigma \ne 0\}}$, and we assume that the minimum over the empty set is equal to $\infty$. 	
It can be shown~\cite{kalachevTwosidedRobustlyTestable2023} that for $i=D-2$ in~\eqref{eq:cobound} the elements $c' = \delta^i c$ range over all non-zero elements of  $\cC_1\boxplus \dots\boxplus \cC_D$, and the minimum in the denominator  is taken over all possible values $\sum_{i\in [D]} \abs{a_i}_i$ such that $c' = \sum_{i\in[D]} a_i$ with $a_i\in \cC^{(i)}$.
Taking into account~\eqref{eq:prod-exp2}, this implies  
\[
\rho(\cC_1,\dots,\cC_D) = \frac{1}{n}\eta^{D-2}(X;\cF).
\]
When all the component codes $\cC_1,\dots,\cC_D$ are the repetition codes $\mathrm{Rep}_n = \{0\dots0,1\dots 1\}\subseteq \F_2^n$, then all local codes $\cF_\sigma$ are also repetition codes and $\cF_\sigma\cong \F_2$, $\sigma\in X$. Thus the constructed above cochain complex with coefficients in $\cF$ is isomorphic to the standard cochain complex over $\F_2$ for $X$ and the $i$-th coboundary expansion $\eta^i(X;\cF)$ coincides with the standard $i$-th coboundary expansion $h^i(X) = \eta^i(X;\F_2)$ (also known as \emph{$i$-th Cheeger constant}) of $X$ (e.g., see~\cite[Definition~2.9]{dotterrerCoboundaryExpanders2012}).
Thus the product expansion factor can be seen as a~natural generalization of the Cheeger constant of $X=X(K_{n,\dots,n})$ for the last non-trivial $i = D-2$, and it is shown in~\cite[Proposition~5.7]{dotterrerCoboundaryExpanders2012} that $\eta^{i}(X;\F_2) \ge n/(2^{D} - 1)$, which implies that:
\[ \rho(\mathrm{Rep}_n,\dots,\mathrm{Rep}_n) \ge \frac{1}{2^{D} - 1}. \]

\section{Proof of the Main Result}

\subsection{Locally Testable Codes}
There are several formal definitions of LTCs~\cite{Goldreich:2010}. Here, we adopt a~rather strong form of local testability (e.g., see~\cite{Leverrier:2021a}). We say that a~linear code~$\mathcal{C} \subseteq \mathbb{F}_q^n$ is {\emph{$(\Delta,s)$-locally testable}} if it has a~parity-check matrix $H\in \F_q^{m\times n}$ with rows of weight at most $\Delta$ such that for any vector $x \in \mathbb{F}_q^n$ we have
\[ \frac{1}{m} | H x | \geqslant \frac{s}{n} d(x, \mathcal{C}), \]
where $d(x, \mathcal{C}) := \min_{c\in\cC} d(x,c)$, and we denote by $d(\cdot,\cdot)$ and $|\cdot|$ the Hamming distance and the Hamming weight.
The parameters $\Delta$ and $s$ are positive real numbers called the~\emph{locality} and \emph{soundness}, respectively. 

The next theorem shows that one can construct good LTCs of arbitrary length from the good LTCs proposed by Dinur~et al.~\cite{Dinur:stoc2022}. We use the fact that the sequence of lengths in this family of good LTCs grows approximately as a~geometric progression, and show that any given code length can be obtained by duplicating codewords and padding with zeros.

\begin{theorem}\label{th:LTC}
    For every  $R\in (0,1)$ there exist constants $s>0$, $\Delta>0$, $\delta>0$ such that for each $n\in \NN$ there is  a~$(\Delta,s)$-locally testable $[n,k,d]$ code for some ${k\ge Rn}$, ${d\ge\delta n}$ defined by a~$\Delta$-limited parity-check matrix with $m$ rows, where\footnote{This additional technical condition is used later to simplify the proof of the main result.} $m\in (n/2,n]$.
\end{theorem}

\subsection{Extendable and Inner-generated Sets}\label{sec:ext-sets}

Given the grid $[n]^D$, we consider the set of \emph{axis-parallel lines} $\mathcal{L}(n,D):=\bigcup_{i=1}^D\mathcal{L}_i(n,D)$, where $\mathcal{L}_i(n,D)$ are the lines in the \emph{$i$-th direction}. If we have a~collection of codes $\cC_1,\dots, \cC_D\subseteq \F_q^n$, then to a~line $\ell\in\cL_i(n,D)$ we can assign the code 
\[
\cC_\ell:= \{c\in \F_q^{[n]^D} \mid  \supp c \subseteq \ell, c|_\ell\in \cC_i\}.
\]

Let $\cC_1,\dots,\cC_D\subseteq \F_q^n$, $M\subseteq [n]^D$. By $L(M)$ denote the set $\fbr{\ell\in\cL(n,D)\mid \ell\subseteq M}$ of all axis-parallel lines contained in~$M$. 
\begin{definition}[inner-generated sets]
A~set $M \subseteq [n]^D$ is called \emph{inner-generated} for a~code $\cC_1\boxplus\ldots\boxplus \cC_D$ if for each of its codewords $c$ with  $\supp c \subseteq M$ we have $c\in \sum_{\ell\in L(M)}{\cC_\ell}$, i.e., it can be represented as a~sum of codewords along the lines contained in~$M$.    
\end{definition}

\begin{figure*}[hbt]
    \centering
    \begin{tikzpicture}[
    scale=0.8,
    grid/.style={draw, black!70!white, very thin, step=1},
    setM/.style={fill=red, fill opacity=0.4},
    codeword/.style={fill=black!40!white},
    dualcodeword/.style={fill=blue!40!white},
    lline/.style={draw,very thick,black},
    lcircle/.style={black},
]

\begin{scope}[xshift=5cm,scale=0.7]
    \begin{scope}
        \fill[setM] (0,0) rectangle (3,1) (1,2) rectangle (3,3) (0,1) rectangle (1,2) (2,1) rectangle (3,2);
        \node[above] at (1.5, -1.2) {The set $M$};
        \node[above] at (1.5, -2.2)  {(a)};
        \draw[grid] (0,0) grid (3,3);
        
        \filldraw[lcircle] (0.5,1.5) circle (0.1cm);
        \filldraw[lcircle] (1.5,2.5) circle (0.1cm);
        \draw[lline] (0.2,0.5)--(2.8,0.5);
        \draw[lline] (2.5,0.2)--(2.5,2.8);
    \end{scope}

    \begin{scope}[xshift=6cm]
        \fill[codeword] (0,2) rectangle (3,3); 
        \node[above] at (1.5, -1.2) {$c_1 \in \mathcal{C}^{(1)}$};
        \draw[grid] (0,0) grid (3,3);
        \node at (3.75, 1.4) {$+$};
    \end{scope}
    
    \begin{scope}[xshift=10.5cm]
        \fill[codeword] (0,0) rectangle (1,3); 
        \node[above] at (1.5, -1.2) {$c_2 \in \mathcal{C}^{(2)}$};
        \draw[grid] (0,0) grid (3,3);
        \node[above] at (1.5, -2.2)  {(b)};
        \node at (3.75, 1.4) {$=$};
    \end{scope}
    
    \begin{scope}[xshift=15cm]
        
        \fill[codeword] (0,0) rectangle (1,2); 
        \fill[codeword] (1,2) rectangle (3,3); 
        \node[above] at (1.5, -1.2) {$c \in \cC_1\boxplus\cC_2$};
        \draw[grid] (0,0) grid (3,3);
        \filldraw[lcircle] (0.5,1.5) circle (0.1cm);
        \filldraw[lcircle] (1.5,2.5) circle (0.1cm);
        \draw[lline] (0.2,0.5)--(2.8,0.5);
        \draw[lline] (2.5,0.2)--(2.5,2.8);
    \end{scope}

    \begin{scope}[xshift=21cm]
        \fill[setM] (0,0) rectangle (3,1) (1,2) rectangle (3,3) (0,1) rectangle (1,2) (2,1) rectangle (3,2);
        
        \fill[dualcodeword] (1,2) rectangle (2,3); 
        \node[above] at (1.5, -1.2) {$w$};

        \draw[grid] (0,0) grid (3,3);
        \filldraw[lcircle] (0.5,1.5) circle (0.1cm);
        \filldraw[lcircle] (1.5,2.5) circle (0.1cm);
        \draw[lline] (0.2,0.5)--(2.8,0.5);
        \draw[lline] (2.5,0.2)--(2.5,2.8);
        \node[above] at (1.5, -2.2)  {(c)};
    \end{scope}
\end{scope}
\end{tikzpicture}
    \caption{(a) Example of a~set $M\in [3]^2$ that is not inner-generated for the code $\cC_1\boxplus \cC_2$ where $\cC_1=\cC_2$ are the repetition $[3,1,3]$ codes. The points $(1,2)$ and $(2,1)$ shown as black dots do not belong to the lines contained in $M$ shown as the black lines. Therefore the codeword $c\in \cC_1\boxplus \cC_2$ shown in Subfigure (b) with non-zero elements on $M$ cannot be represented as a~sum of codewords along the lines contained in $M$. Subfigure (c) illustrates the idea of Proposition~\ref{pr:1}: the word $w$ on $M$ with one non-zero bit (the blue cell) cannot be extended to a~codeword from the dual code $\cC_1^\perp\otimes\cC_2^\perp$ since it does not satisfy the check $c$ though satisfies all the checks along the lines in $M$, thus $M$ is not extendable in $\cC_1^\perp\otimes\cC_2^\perp$.}
    \label{fig:inner-gen}
\end{figure*}

Since $\cC_1\boxplus\ldots\boxplus \cC_D$ is the dual of the code 
\[\cC^\top := \cC_1^\perp\otimes\ldots\otimes \cC_D^\perp,\] we can view its codewords as checks for $\cC^\top$. With this interpretation in mind, the set of bit positions $M\subseteq [n]^D$ is inner-generated \Iff every check $z$ of $\cC^\top$  with $\supp z \subseteq M$ can be obtained as a~linear combination of the elements from $\cC_\ell$, $\ell\in L(M)$. The last condition is, in turn, equivalent to the condition that every local codeword $c_M \in \F_q^M$ satisfying only local checks $c_M|_\ell \in \cC_i^\perp$ for all $\ell\in L(M)\cap \cL_i(n, D)$ and $i\in [D]$, can be \emph{extended} to a~global codeword $c\in \cC^\top$ such that $c|_M = c_M$. This motivates us to give the following definition, which is dual to the notion of an~inner-generated set. 

\begin{definition}[extendable sets]
We call a~set $M$ \emph{extendable} for a~code $\cC_1\otimes\ldots\otimes \cC_D$ if every local codeword $c_M \in \F_q^M$ satisfying $c_M|_\ell \in \cC_i$ for all $\ell\in L(M)\cap \cL_i(n, D)$ and $i\in [D]$ (i.e., the local checks for $\ell \in L(M)$), can be extended to a~global codeword $c\in \cC_1\otimes\ldots\otimes \cC_D$ such that $c|_M = c_M$.    
\end{definition}

Thus, we have the following fact:
\begin{proposition}\label{pr:1}
    A~set $M\subseteq [n]^D$ is inner-generated for a~code $\cC_1\boxplus\ldots\boxplus \cC_D$ \Iff $M$ is extendable in $\cC_1^\perp\otimes\ldots\otimes \cC_D^\perp$.
\end{proposition}

\begin{definition}[maximally extendable product code]\label{def:ME}
We say that a~product code $\cC=\bigotimes_{i\in [D]} \cC_i \subseteq \F_{2^t}^{[n]^D}$ is \emph{maximally extendable} if for every other product code 
\[\cC' = \bigotimes_{i\in [D]} \cC'_i \subseteq \F_{2^{t'}}^{[n]^D}\] with $\dim \cC_i = \dim \cC_i'$, $i\in[D]$, when $M$ is extendable in $\cC'$, it is also extendable in $\cC$. 
\end{definition}

\begin{remark}
    The idea of maximally extendable codes is similar to that of \emph{maximally recoverable} (MR) codes~\cite{chenMaximallyRecoverableProperty2007, gopalanExplicitMaximallyRecoverable2014}, where in Definition~\ref{def:ME} one replaces extendable sets by information sets\footnote{An~\emph{information set} of a~linear code is a~set $M$ of coordinate positions whose symbols can be chosen freely and which then uniquely determine the entire codeword, i.e., the rest of the codeword can be \emph{recovered} from them.}. However, there is a~technical difference between these two concepts. Clearly, any information set of a~product code is extendable and does not contain any line, hence in a~maximally extendable code this set is also an~information set; therefore, a~maximally extendable code has all information sets that product codes with the same rates of component codes can have; thus it is also MR. However, it is an~open question whether any MR product code is maximally extendable. Although from any set of lines it is possible to choose an~information set for the local code defined by constraints along these lines, MR property does not guarantee that this set will be a~subset of an~information set of the global code defined by \emph{all} constraints.
\end{remark}

\subsection{Proof Outline of Theorem~\ref{th:rand-expanding}}\label{sec:outline}

It is known~\cite{kalachevTwosidedRobustlyTestable2023} that the expansion of $\cC_1 \otimes \cC_2$ is almost the same as its robust testability, which, in turn, is much easier to prove for the product of a~good code $\cC_1$ and an~expander LDPC code $\cC_2$ (originally, this result was formulated in terms of smooth codes \cite{Ben-Sasson:2006}). But when we need to prove this property for both $\cC_1 \otimes \cC_2$ and $\cC_1^\bot \otimes \cC_2^\bot$, which is required in many existing constructions of good qLDPC codes, we cannot use LDPC codes because their dual codes have small minimum distance and, hence, $\cC_1^\bot \otimes \cC_2^\bot$ does not have good expansion. To build such pairs of codes, the most natural approach is to prove that pairs of random codes are product-expanding; this approach was used in \cite{Panteleev&Kalachev:stoc2022,Dinur:decoders,kalachevTwosidedRobustlyTestable2023}. The existing proofs for pairs of random codes are considerably more complicated than the proof for the case where one of the codes is a~smooth code.

For $D > 2$, proving the good product expansion of random codes is even more challenging. The ideas used in the proof for $D=2$ cannot be generalized in a~straightforward manner for the case $D>2$. 
And even for a~collection of expander LDPC codes, the proof for $D=2$ cannot be generalized directly. 
However, the idea of using codes with some ``smooth'' properties can be applied to the collection of LTCs. 
The important property of LTCs that indirectly helps here is that, loosely speaking, if $\cC_1,\dots,\cC_D$ are LTCs, then the code $\cC:=\cC_1\boxplus\dots\boxplus\cC_D$ is also an~LTC. 
This property allows us to estimate the distance from a~vector to the code $\cC$ by the Hamming weight of the syndrome. 
It also helps to prove that the collection of good LTCs is product-expanding (Lemmas~\ref{lemma:prod-ltc}, \ref{lemma:LTC-exp}). In addition, some form of rate adaptation is required to obtain examples of product-expanding collections of codes with arbitrary rates (Lemma \ref{lemma:prodexp-exists}).

However, the duals of LTCs are codes with small distances, so the tuple of dual codes of these LTCs cannot have good product expansion. To address this, we give a criterion for good product expansion in terms of extendable sets in the product of dual codes (Proposition~\ref{pr:1} and Lemmas~\ref{lemma:inner-generated}, \ref{lemma:closure-size}, \ref{lemma:ig-exp}) and introduce the universal property of a~collection of codes that inherits all extendable sets from all other collections\footnote{It is similar to the idea of maximally recoverable (MR) codes~~\cite{chenMaximallyRecoverableProperty2007, gopalanExplicitMaximallyRecoverable2014}.} (Definition~\ref{def:ME}).
If a~collection of codes has this universal property, then it inherits the product expansion from a~collection of LTCs with the same parameters (Lemma \ref{lemma:universal-exp}).
Then, using the Schwartz-Zippel Lemma, we prove that a~random collection of codes over a~large enough field has this universal property with high probability (Lemma \ref{lemma:rand-universal}). 
Hence, if the field is large enough, then, with high probability, a~collection of codes has good product expansion.

However, to implement this plan, we need to have LTCs of arbitrary length $n$ and dimension $k$. The existing constructions of LTCs \cite{Dinur:stoc2022,Panteleev&Kalachev:stoc2022,Leverrier:focs2022,Lin2022:losslessLTC} can give LTCs only for some specific values of $n$ and $k$, so we need to perform some rate adaptation. To do this, we first prove that we can construct good LTCs of arbitrary length\footnote{It is still unknown how to construct LTCs with arbitrary dimension.} and rate bounded from below (Theorem \ref{th:LTC}). 
The second step is to prove that if we reduce the dimension of the codes in a~collection by taking subspaces, the product expansion does not become much worse (Lemma~\ref{lemma:subcode-exp}). In this way, we adjust the dimension of the codes without changing the length, which allows us to obtain Lemma~\ref{lemma:prodexp-exists} that states collections of product-expanding codes exist for arbitrary lengths and arbitrary dimensions of the codes.

We believe that LDPC codes derived from lossless expanders (e.g., random constructions) could potentially replace LTCs in our proof. However, proving product expansion for such codes is more challenging than for LTCs. While this alternative strategy avoids rate adaptation and may yield better constants, it remains an~open direction for future work.

\subsection{Operations Preserving Product Expansion}

Before proceeding to the main steps of our proof, as outlined in Section \ref{sec:outline}, we first establish a~few auxiliary results. A~key part of our strategy involves showing the existence of product-expanding codes for arbitrary rates. To achieve this, we will start with a~known construction of LTCs and then take subcodes to adjust the rates. The following lemmas provide the necessary tools, ensuring that the expansion property does not degrade too much in this process.

In this section, we will use the following auxiliary results from \cite{kalachevTwosidedRobustlyTestable2023}.
\begin{extlemma}[Lemma 6 in \cite{kalachevTwosidedRobustlyTestable2023}]\label{lemma:intersection}
    For linear codes $\cC_1,\cC_2,X,Y\subseteq \F_q^n$ we have: 
    $$(X\otimes Y)\cap (\cC_1\boxplus \cC_2)=(X\cap \cC_1)\otimes Y+X\otimes (Y\cap \cC_2).$$
\end{extlemma}

\begin{extlemma}[Lemma 11 in \cite{kalachevTwosidedRobustlyTestable2023}]\label{lemma:subset-exp}
Let $\cC=(\cC_i)_{i\in [D]}$ be a~$\rho$-product-expanding collection of codes $\cC_i\subsetneq \F_q^{n_i}$. Then each subcollection $\cC_I=(\cC_i)_{i\in I}$, $I\subseteq [D]$, is also $\rho$-product-expanding.
\end{extlemma}

\begin{lemma}\label{lemma:subcode-exp}
Consider codes $C_1,\dots,C_D\subseteq\F_q^n$, and a~subcode $C_1' \subseteq C_1$. Then 
\[\rho(C_1',C_2,\dots,C_D) \ge \frac{\rho(C_1,C_2,\dots,C_D)}{1+\rho(C_2,\dots,C_D)^{-1}}
.\]
\end{lemma}
\begin{proof}
    Let $\rho:=\rho(C_1,\dots,C_D)$, $\rho_r:=\rho(C_2,\dots,C_D)$.  Consider $C_r := C_2\boxplus\cdots\boxplus C_D$. Denote the length of the code $C_r$ as $n_r=n^{D-1}$. Consider an arbitrary codeword $x \in C_1' \boxplus C_r$. Since $x \in C_1 \boxplus C_r$, there exist words $x_i \in C^{(i)}$, such that $\rho\sum |x_i|_i n \le |x|$. 

    Let $A \subseteq \F_q^{n}$ be such that $A \cap C_1' = \{0\}$, $A + C_1' = C_1$. Then $x_1 = x_1' + x_A$ for some $x_1' \in C_1' \otimes \F_q^{n_r}$ and $x_A \in A \otimes \F_q^{n_r}$. Since $x, x'_1, x_j\in C'_1\boxplus C_r$ for $2\le j\le D$, we have \[
        x_1=x-\sum_{j=2}^Dx_j\in C'_1\boxplus C_r,\quad x_A=x_1-x'_1\in C'_1\boxplus C_r,
    \]
    which means $x_A \in (A \otimes \F_q^{n_r}) \cap (C_1' \boxplus C_r)$. By Lemma \ref{lemma:intersection} we have
    \begin{align*}
    (A \otimes \F_q^{n_r}) \cap (C_1' \boxplus C_r) 
    &= (A \cap C_1') \otimes \F_q^{n_r}\! +\! A \otimes (\F_q^{n_r} \cap C_r) \\
    &= A \otimes C_r \subseteq \F_q^{n} \otimes C_r.
    \end{align*}
    Therefore, we get a~decomposition $x_A = x_2' + \cdots + x_D'$ with $x_2' \in C^{(2)}, \dots, x_D' \in C^{(D)}$ such that \[
        \rho_r \sum_{i=2}^D |x_i'|_i n \le |x_A| \le |x_A|_1 n \le |x_1|_1 n.
    \] 
    Put $c_1 := x_1'$, $c_i := x_i + x_i'$ for $i=2,\dots,D$. Then $c_1 \in C'^{(1)}$, $c_i \in C^{(i)}$ for $i=2,\dots,D$ and
    \(x = \sum_{i=1}^D c_i\). Furthermore, 
    \begin{align*}
    \sum_{i=1}^D |c_i|_i n &\le \sum_{i=2}^D |x_i'|_i n + |x_1'|_1 n + \sum_{i=2}^D |x_i|_i n\\
    &\le \frac{1}{\rho_r}|x_1|_1n + \sum_{i=1}^D |x_i|_i n\\
    &\le \frac{1}{\rho_r \rho}|x| + \frac{1}{\rho}|x| = \frac{1 + \frac{1}{\rho_r}}{\rho}|x|.
    \end{align*}
    Hence, $\rho(C_1',C_2,\dots,C_D) \ge \frac{\rho}{1+\frac{1}{\rho_r}}$, which completes the proof. 
\end{proof}

\begin{corollary}\label{sled:subcodes}
Let $C_1,\dots,C_D\subseteq \F_q^n$ be linear codes, and $C_i' \subseteq C_i$ for $i \in [D]$. Then 
\[\rho(C_1',C_2',\dots,C_D') \ge 2^{-D}\bigrbr{\rho(C_1,C_2,\dots,C_D)}^{2^D}.\]    
\end{corollary}
\begin{proof}
    Consider 2 cases. 
    If one of the codes $C_i$ has rate $1$, then $\rho(C_1,C_2,\dots,C_D)=1/n$; thus, we have
        \begin{align*}            
            \rho(C'_1,\dots,C'_D)\ge 1/n^{D-1}&>2^{-D}n^{-2^D}\\
            &=2^{-D}\bigrbr{\rho(C_1,C_2,\dots,C_D)}^{2^D}.
        \end{align*}
        Otherwise, by Lemma \ref{lemma:subset-exp} we have \[
            \rho(C_2,\dots,C_D) \ge \rho(C_1,\dots,C_D),
        \] hence by Lemma~\ref{lemma:subcode-exp} and taking into account $\rho\le 1$ we have 
        \begin{align*}
            \rho(C_1',C_2,\dots,C_D) 
            &\ge \frac{\rho(C_1,C_2,\dots,C_D)}{1+\rho(C_2,\dots,C_D)^{-1}}\\
            &\ge \frac12\rho(C_1,C_2,\dots,C_D)^2.
        \end{align*}
        Applying this estimation sequentially for each of $D$ codes, we obtain the required inequality.
\end{proof}

\subsection{Products of LTCs are Expanding}

As outlined in Section \ref{sec:outline}, we now execute the first major step of our proof: showing that a collection of locally testable codes (LTCs) is product-expanding. The key insight here is that the local testability of the component codes allows us to bound the distance to a~code by the Hamming weight of the syndrome.

We say that $(\alpha_l,\alpha_h)$ is a~\emph{soundness range} for a~matrix $H\in \F_q^{m\times n}$ if, for any $x\in \F_q^n$, we have 
\[\alpha_l d(x,\ker H)\le |H x|\le \alpha_h d(x,\ker H).\] 
We are interested in the case when $H$ is a~parity-check matrix of an~LTC, and hence, $\alpha_h/\alpha_l=\Theta(1)$.
It is easy to see that the matrices $I_t\otimes H$, $H\otimes I_t$, and $H$ have the same soundness ranges, where $I_t$ is the $t\times t$ identity matrix, $t\in \NN$.

\begin{lemma}\label{lemma:inv-supp}
Let $(\alpha_l, \alpha_h)$ be a~soundness range for a~matrix $H \in \F_q^{m \times n}$, with $d(\ker H) \geq \delta n$. Then for any $A \subseteq [m]$, there exists a~set $B \subseteq [n]$ and a~linear map $\phi_{H,A}:\F_q^A\cap \im H\to \F_q^B$ such that $|B| \leq c|A|$, where $c = \frac{6}{\delta \alpha_l}$, and $H \phi_{H,A}(y)=y$ for all $y \in \im H$ with $\supp y \subseteq A$.
\end{lemma}

\begin{proof}
Fix some subset $A \subseteq [m]$, and let $a = |A|$. If $a \geq n/c$, then it suffices to choose $B = [n]$, and the lemma is proved. Thus, we can assume that $a < n/c$. Consider a~basis $y_1, \dots, y_s$ of the vector space \[
    \F_q^A\cap \im H=\{y \in  \im H\ \mid \supp y \subseteq A\}.
\] 
For each $i\in [s]$, choose a~vector $x_i \in \F_q^n$ of minimal weight such that $H x_i = y_i$. Define \[
    B := \bigcup_{i=1}^s \supp x_i,\qquad \phi_{H,A}(y_i):=x_i
\]
for $i\in [s]$ and extend it by linearity on the whole subspace $\F_q^A\cap \im H$. By definition, we have $\im \phi_{H,A}\subseteq \F_q^B$.

Since $(\alpha_l, \alpha_h)$ is a~soundness range for $H$, for all $i\in [s]$ we have $\abs{x_i} \le a/\alpha_l$. 
To prove the lemma, it suffices to show that $|B| \leq 2a/\alpha_l < ac$. 
Assume the converse $|B| > 2a/\alpha_l$, and find the smallest $s'$ such that $|\bigcup_{i=1}^{s'} \supp x_i| > 2a/\alpha_l$. 
Consider $B' := \bigcup_{i=1}^{s'} \supp x_i$. It is easy to see that $s' > 2$ and $2a/\alpha_l < |B'| \leq 3a/\alpha_l$. 
Let us choose uniformly at random a~vector from the linear span $\abr{x_1, \dots, x_{s'}}$. 
It is easy to see that the expected value of its Hamming weight is \[
    |B'|(1-1/q) \geq |B'|/2 > a/\alpha_l,
\]
so there exists a~vector $x \in \abr{x_1, \dots, x_{s'}}$ such that 
\[
a/\alpha_l < |x| \leq |B'| \leq 3a/\alpha_l < \frac{3n}{c\alpha_l} \leq \delta n/2 \leq d(\ker H)/2.
\]
Therefore, $d(x, \ker H) = |x|$, and then $|H x| \geq \alpha_l |x| > a$, which contradicts $\supp H x \subseteq A$. Thus, $|B| \leq 2a/\alpha_l < ac$, and the lemma is proved.
\end{proof}

We will use the following obvious property of the product expansion.

\begin{lemma}\label{lemma:add-zero-code}
    For codes $\cC_1,...,\cC_D\subseteq \F_q^n$ and the zero code $\mathbf{0}\subseteq \F_q^n$ we have \[
    \rho(\mathbf{0},\cC_1,...,\cC_D)= \rho(\cC_1,...,\cC_D).
    \]
\end{lemma}

The following lemma plays the~central role in this section. It proves that a~collection of LTCs is product-expanding. The argument proceeds via induction on $D$, the number of component codes. The core technical tool is Lemma~\ref{lemma:inv-supp}, which, using the soundness property of LTCs, allows us to construct a~``preimage'' of a~syndrome with controlled support size. This fact is utilized in the proof to give a~lower bound for the product-expansion.

\begin{lemma}\label{lemma:prod-ltc}
Consider codes $\cC_1,\dots,\cC_D\subseteq \F_q^n$ each having minimum distance at least $\delta n$ and a~parity-check matrix $H_i$ with soundness range $(\alpha_l,\alpha_h)$. Then for all $n$ we have that $\rho(\cC_1,\dots,\cC_D)$ is bounded from below by a~positive function $f(D, \alpha_l, \alpha_h, \delta)$ that does not depend on $n$.
\end{lemma}

\begin{proof}
We will prove the lemma by induction on the dimension $D$. For $D=1$ we have $\rho(\cC_1)=d(\cC_1)/n \ge \delta$, so we set $f(1,\alpha_l,\alpha_h,\delta)=\delta$. Now, we let $D\ge 2$ and assume that the lemma holds for $D-1$.

Let $\cC:=\cC_1\boxplus \dots\boxplus \cC_D$. 
Consider an~arbitrary codeword $x\in \cC$ and its syndrome $s:=(I_n\otimes H_2\otimes \dots\otimes H_D) x$ for the $(D-1)$-dimensional code $\mathbf{0}\boxplus \cC_2\boxplus\dots\boxplus \cC_D$, where $I_n$ is the $n\times n$ identity matrix. 
Then $|s|\le \alpha_h^{D-1}|x|$. 
On the other hand, since
\[(H_1\otimes I_{m_2}\otimes \cdots \otimes I_{m_D})s=(H_1\otimes H_2\otimes \dots\otimes H_D) x=0,\]
we have $s\in \cC_1\otimes \F_q^{m_2\times \dots\times m_D}$ where $m_i$ is the number of rows of the matrix $H_i$ for $i\in [D]$. 
Hence, $\supp s$ is covered by no more than $|s|/d(\cC_1)\le|s|/(\delta n)$ lines in the first direction. Thus $\supp s\subseteq [n]\times A$, where $A\subseteq [m_2]\times \dots\times [m_D]$ and $|A|\le |s|/(\delta n)$.

Let $A_1:=A$. 
Sequentially applying Lemma \ref{lemma:inv-supp} for $i=2,\dots,D$ to the set $A_{i-1}$ and the matrix 
\[
    \hat H_{i}:=I_n^{\otimes (i-2)}\otimes H_{i}\otimes I_{m_{i+1}}\otimes \cdots\otimes I_{m_D},
\] 
each time we get a~set $A_{i}$ and a~map $\phi_{\hat H_i, A_{i-1}}$. 
Finally, we put \[
    B:=A_D,\qquad \phi:=\phi_{\hat H_D,A_{D-1}}\circ \cdots\circ\phi_{\hat H_2, A_1}.
\] 
Then $\phi$ is linear and for all $y\in \im (H_2\otimes \dots \otimes H_D)$ with $\supp y\subseteq A$ we have $\supp \phi(y)\subseteq B$ and \[
    (H_2\otimes \cdots\otimes H_D) \phi(y)=\hat{H}_2\ldots \hat H_D \phi(y)=y.
\] 
Since we applied Lemma \ref{lemma:inv-supp} to matrices with soundness range $(\alpha_l,\alpha_h)$, we have $|B|\le c^{D-1}|A|$, where $c=\frac{6}{\delta \alpha_l}$.

Since $s\in \cC_1\otimes \im (H_2\otimes \dots \otimes H_D)$, the vector \[a^{(1)}:=(I_n\otimes \phi)s\] is well-defined, $a^{(1)}\in \cC_1\otimes \F_q^B\subseteq \cC^{(1)}$, and $|a^{(1)}|_1\le |B|$. Also, by the definition of $\phi$, we have \[(I_n\otimes H_2\otimes\cdots\otimes H_D)a^{(1)}=s.\]

Now note that the vector $x':=x-a^{(1)}$ lies in the code $\mathbf{0}\boxplus \cC_2\boxplus\dots\boxplus \cC_D$, since \begin{multline*}
    (I_n\otimes H_2\otimes \dots\otimes H_D)x'=(I_n\otimes  H_2\otimes \dots\otimes H_D)x\\
    -(I_n\otimes H_2\otimes \dots\otimes H_D)a^{(1)}=s-s=0.
\end{multline*}
By Lemma \ref{lemma:add-zero-code} and the inductive hypothesis we have \[
    \rho(\mathbf{0},\cC_2,\dots, \cC_D)
    =\rho(\cC_2,\dots, \cC_D)
    \ge f(D-1,\alpha_l,\alpha_h,\delta).
\]
Hence, $x'=\sum_{j=2}^D a^{(j)}$ for some $a^{(j)}\in \cC^{(j)}$, $2\le j\le D$ such that \[
    n f(D-1,\alpha_l,\alpha_h,\delta)\sum_{j=2}^D |a^{(j)}|_j\le |x'|.
\]
We have: \[
    |a^{(1)}|
    \le |B|n
    \le c^{D-1}|A|n
    \le \frac{c^{D-1}|s|}{\delta}
    \le c^{D-1}\alpha_h^{D-1}\delta^{-1}|x|.
\] \[  
    |x'|=|x-a^{(1)}|
    \le |x|+|a^{(1)}|
    \le \rbr{1+c^{D-1}\alpha_h^{D-1}\delta^{-1}}|x|.
\]
Therefore, \[
    x=a^{(1)}+x'=\sum_{j=1}^D a^{(j)}, \quad 
    a^{(j)}\in \cC^{(j)}\mbox{ for }j\in [D].
\]
We can estimate $N:=\sum_{j=1}^D |a^{(j)}|_j$ as follows: 
\begin{align*}
    N &\le |B|+\frac{|x'|}{nf(D-1,\alpha_l,\alpha_h,\delta)}\\
    &\le \frac{(c\alpha_h)^{D-1}}{\delta  n}|x| + \frac{\rbr{1+\frac{(c\alpha_h)^{D-1}}{\delta}}|x|}{nf(D-1,\alpha_l,\alpha_h,\delta)}.
\end{align*}
We can simplify this estimate using that \[
c\ge 1,\quad \delta\le 1,\quad \alpha_h\ge 1,\quad f(D-1,\alpha_l,\alpha_h,\delta)\le\delta\le 1,
\]
and substitute $c$ from Lemma~\ref{lemma:inv-supp}:\[
    N\le \frac{|x|}{n}\cdot\frac{3(c\alpha_h)^{D-1}}{\delta\!\cdot\! f(D-1,\alpha_l,\alpha_h,\delta)}
    =\frac{3\rbr{\frac{6\alpha_h}{\alpha_l\delta}}^{D-1}}{\delta\!\cdot\! f(D-1,\alpha_l,\alpha_h,\delta)}\cdot \frac{|x|}{n}.
\]
Thus, we can take $f(D,\alpha_l,\alpha_h,\delta):=\frac{\delta \cdot f(D-1,\alpha_l,\alpha_h,\delta)}{3\rbr{\frac{6\alpha_h}{\alpha_l\delta}}^{D-1}}$, and the lemma is proved.
\end{proof}
Note that in the proof we did not use the soundness range for the code $\cC_1$; therefore $\cC_1$ can be an~arbitrary code with distance at least $\delta n$.

The following statement directly follows from the definition of LTCs, the $\Delta$-limited matrix, and the soundness range.

\begin{lemma}\label{lemma:LTC-exp}
    If a~code $\cC\subseteq \F_q^n$ is $(\Delta, s)$-locally testable with a~$\Delta$-limited parity-check matrix $H$ of size $m\times n$, where $n/2\le m\le n$, then $(s/2,\Delta)$ is a~soundness range for the matrix $H$.
\end{lemma}

After establishing that collections of specific LTCs are product-expanding, we now use this to show that such collections exist for \emph{any} set of desired rates below one. As mentioned in the outline, this requires some rate adaptation. The following lemma achieves this by combining our main result for LTCs (Lemma~\ref{lemma:prod-ltc}) with the existence of LTCs of arbitrary length (Theorem~\ref{th:LTC}) and the preservation of expansion under taking subcodes (Corollary~\ref{sled:subcodes}). The following lemma provides the final existence result that our main theorem will draw upon.

\begin{lemma}\label{lemma:prodexp-exists}
    For arbitrary $r_1,\dots,r_D\in (0,1)$ there exists $\rho>0$ such that for any $n\in \NN$ and $k_i\in\NN$ such that $k_i\le r_i n$, there exist codes $\cC_i\subseteq\F_2^n$, $i\in [D]$, such that $\dim \cC_i=k_i$ and 
    \[\rho(\cC_1,\dots,\cC_D)\ge \rho.\]
\end{lemma}

\begin{proof}
    Let $r:=\max(r_1,\dots,r_D)$. By Theorem~\ref{th:LTC}, for the rate $r$, choose $s>0,\Delta>0,\delta>0$. Let \[
    \rho:=2^{-D}(f(D, s/2,\Delta, \delta))^{2^D},
    \] where the function $f$ is from Lemma~\ref{lemma:prod-ltc}. Fix $n$ and $k_i$ such that $k_i\le r_in$. 
    Let us show that the required collection of codes exists.
    \begin{enumerate}
        \item By Theorem~\ref{th:LTC}, there exists a~$(\Delta,s)$-LTC $\cC$ of length $n$ such that $\dim \cC\ge rn$, $d(C)\ge\delta n$, with a~$\Delta$-limited parity-check matrix $H$ of size $m\times n$, where $n/2\le m\le n$.
        \item By Lemma~\ref{lemma:LTC-exp}, the pair $(s/2,\Delta)$ is a~soundness range for the matrix $H$.
        \item By Lemma~\ref{lemma:prod-ltc}, with $\alpha_l=s/2,\alpha_h=\Delta$, we have 
        \[\rho(\underbrace{\cC,\dots,\cC}_{D\text{ times}})\ge f(D,s/2,\Delta,\delta)>0.\]
        \item For each code $i\in [D]$, we have $\dim \cC\ge rn\ge k_i$; choose an~arbitrary subcode $\cC_i\subseteq \cC$ of dimension $k_i$. By Corollary~\ref{sled:subcodes}, we have
        \begin{align*}
        \rho(\cC_1,\dots,\cC_D)&\ge 2^{-D}\rho(\cC,\dots,\cC)\\
        &\ge 2^{-D}(f(D, s/2,\Delta, \delta))^{2^D}=\rho.
        \end{align*}
    \end{enumerate}
    The lemma is proved.
\end{proof}

\subsection{Maximally Extendable Product Codes}

In the previous sections we showed that products of LTCs have good expansion. However, as highlighted in our proof outline (Section \ref{sec:outline}), this is insufficient for our goals, as the duals of LTCs have poor distance. To circumvent this obstacle, we now introduce the central notion of this paper: \emph{maximally extendable product codes}. Our strategy is to reformulate product expansion in terms of extendable subsets and then show that random codes possess a~universal property that allows them to inherit good expansion from the LTC-based construction.

The \emph{$\eps$-closure of the set $M \subseteq [n]^D$} is defined as the minimum set $[M]_\eps \subseteq [n]^D$ containing $M$ such that for any line $\ell\in \mathcal{L}(n,D)$, either $\ell\subseteq [M]_\eps$ or $|\ell\cap [M]_\eps|<\eps n$. If $[M]_\eps=M$, then the set $M$ is called \emph{$\eps$-closed}.

\begin{figure*}[hbt]
    \centering
\newcommand{\drawgrid}[3]{%
  \draw[gray!50, thin, step=1] (0,0) grid (6,6);
  \foreach \x/\y in #1 {
    \fill[black, opacity=0.3]
      (\x,\y) rectangle ++(1,1);
  }
  \if\relax\detokenize{#2}\relax
  \else
  \foreach \x/\y in #2 {
    \fill[red, opacity=0.3]
      (\x,\y) rectangle ++(1,1);
  }
  \fi
  \if\relax\detokenize{#3}\relax
  \else
      \foreach \x/\y in #3 {
        \fill[blue, opacity=0.3]
          (\x,\y) rectangle ++(1,1);
      }
  \fi
}
\newcommand{\steparrow}{
    \draw[-{Latex[length=2mm, width=1.5mm]}, line width=0.5pt] (-1.2, 3) -- (-0.3, 3);
}

\begin{tikzpicture}[scale=0.4]

\def\InitialCells{0/5, 4/5, 5/5, 5/4, 1/3, 2/3, 4/2, 4/1, 5/1, 0/0,1/0, 2/0,3/0}
\def\StepA{4/0, 5/0}
\def\StepB{4/3, 4/4, 5/2, 5/3}
\def\StepAB{4/0, 5/0,4/3, 4/4, 5/2, 5/3}
\def\StepC{0/3, 3/3}
\def\StepABC{4/0, 5/0,4/3, 4/4, 5/2, 5/3,0/3, 3/3}
\begin{scope}
    \drawgrid{\InitialCells}{}{};
    \node[below] at (3,0) {$M$};
\end{scope}

\begin{scope}[xshift=7.5cm]
    \steparrow
    \drawgrid{\InitialCells}{\StepA}{};
    \draw[red, thick] (0, 0) rectangle (6, 1);
\end{scope}

\begin{scope}[xshift=15cm]
    \steparrow
    \drawgrid{\InitialCells}{\StepB}{\StepA};
    \draw[red, thick] (4.0, 0.0) rectangle (5.0, 6.0);
    \draw[red, thick] (5.0, 0.0) rectangle (6.0, 6.0);
\end{scope}

\begin{scope}[xshift=22.5cm]
    \steparrow
    \drawgrid{\InitialCells}{\StepC}{\StepAB};
    \draw[red, thick] (0, 3) rectangle (6, 4);
\end{scope}

\begin{scope}[xshift=30cm]
    \steparrow
    \drawgrid{\InitialCells}{}{\StepABC};
    \node[below] at (3,0) {$[M]_{\eps}$};
\end{scope}

\end{tikzpicture}
    \caption{Steps for constructing $\eps$-closure of a~set $M\subseteq [n]^2$ shown for $\eps=1/2$ and $n=6$. At each step we add to the set all lines intersecting with it in more than $\eps n$ points.}
    \label{fig:epsclosure}
\end{figure*}

\newcommand{\epsmax}{\eps_{\max}}

For a~collection of codes $\cC_1,\ldots,\cC_D\subsetneq \F_q^n$ let us define $\epsmax(\cC_1,\ldots,\cC_D)$ as the maximum $\eps>0$ such that each $\eps$-closed set in $[n]^D$ is inner-generated for the code $\cC_1\boxplus\ldots\boxplus\cC_D$. The motivation for this notation comes from Proposition~\ref{pr:1}, which implies that for a~collection $\cC_1,\ldots,\cC_D$ such that $\otimes_{i\in[D]}\cC_i^\perp$ is maximally extendable we have
\begin{equation}\label{eq:expmax-inherit}
\epsmax(\cC_1,\ldots,\cC_D)\ge \epsmax(\cC'_1,\ldots,\cC'_D)
\end{equation}
for any other collection of codes $\cC'_1,\ldots,\cC'_D$ with the same rates and lengths.

As we will prove later, every product code $\cC=\otimes_{i\in[D]}\cC_i$ has good expansion provided that the code $\cC^\top=\otimes_{i\in[D]}\cC_i^\perp$ is maximally extendable.  
However, the proof is not direct and involves a~connection between the product-expansion factor $\rho$ and the parameter $\epsmax$ that we defined. The main idea is to establish this connection in the following few lemmas, and then conclude that the collection of codes $\cC_1,\ldots,\cC_D$ inherits $\epsmax$ from the codes constructed in Lemma~\ref{lemma:prodexp-exists}, and therefore also inherits $\rho$ with some constant loss.

The following lemma shows that $\rho$-product-expansion implies that all $\rho$-closed sets are inner-generated.

\begin{lemma}\label{lemma:inner-generated}
 Let $\cC_1,\dots,\cC_D\subsetneq \F_q^n$, $\rho:=\rho(\cC_1,\dots,\cC_D)$. Then each $\rho$-closed subset $M\subseteq [n]^D$ is inner-generated for the code $\cC:=\cC_1\boxplus\ldots\boxplus\cC_D$.
\end{lemma}
\begin{proof}
    For a~set $M\subseteq [n]^D$ denote $\cC_{L(M)}:=\sum_{\ell\in L(M)}\cC_\ell$.
    Suppose there exists a~$\rho$-closed set $M\subseteq[n]^D$ that is not inner-generated. 
    Then the set $S:=(\cC\cap \F_q^M)\setminus \cC_{L(M)}$ is non-empty.
    Consider a~vector $x\in S$ of minimal weight. 
    Since $\rho=\rho(\cC_1,\dots,\cC_D)$, the vector $x$ can be represented as a~sum $\sum_{i=1}^s x_i$ of $s$ words $x_i\in\cC_{\ell_i}$ along some lines $\ell_i\in\cL(n,D)$ where $s\le |x|/(\rho n)$. 
    Without loss of generality, we can assume that the lines are ordered in such a way that $\ell_i\in L(M)$ for $1\le i\le s'$ and $\ell_i\not\in L(M)$ for $s'<i\le s$. 
    Let \[
        x':=\sum_{i=1}^{s'} x_i,\qquad x'':=x-x'=\sum_{i=s'+1}^s x_i.
    \] 
    Then $\supp x'\subseteq M$, and therefore $\supp x''\subseteq M$. 
    Since ${x\not\in\cC_{L(M)}}$, $x'\in \cC_{L(M)}$ we have $x''\not\in \cC_{L(M)}$. 
    Hence, $x''\in S$. 
    Since the set $M$ is $\rho$-closed and $\ell_i\not\in L(M)$ for $s'<i\le s$, we have $|\ell_i\cap M|<\rho n$. 
    Taking into account that $\supp x''\subseteq M$, we obtain
    \begin{align*}
    |x''|\le\sum_{i=s'+1}^{s}\bigabsbr{x_i|_{M}} 
    &\le \sum_{i=s'+1}^{s}|\ell_i\cap M|\\
    &< (s-s')\rho n
    \le s\rho n\le |x|. 
    \end{align*}
    Thus, we have a contradiction with the assumption that $x$ has the minimal weight in the set $S$. This completes the proof.
\end{proof}

\begin{lemma}\label{lemma:closure-size}
 For any $\eps\in (0,1]$ and $D\in\NN$, there exists a~constant $c>0$ such that $|[M]_\eps|\le c|M|$ for any $n\in\NN$ and $M\subseteq [n]^D$.
\end{lemma}
\begin{proof}
 Consider a~set $M$. We are going to construct a~set $M'$ that contains $[M]_\eps$. Set $\eps':=\eps/2^D$. For a~subset $I\subseteq [D]$, let $M_I$ denote the set of cells covered by the $|I|$-dimensional hyperplanes that are parallel to the axes corresponding to the indices in $I$, containing at least $(\eps' n)^{|I|}$ elements of $M$. Then $M':=\bigcup_{I\subseteq [D]} M_I$ contains $M$. 
 
 Let us show that $[M]_\eps\subseteq M'$. To this end, it suffices to verify that the set $M'$ is $\eps$-closed. Consider an~arbitrary line $\ell\in \cL(n,D)$ such that $\ell\not\subseteq M'$. Suppose that $|\ell\cap M_I|\ge \eps'n$ for some $I\subseteq [D]$. Then in an~$(|I|+1)$-dimensional hyperplane $P$ that passes through the line $\ell$ in directions from the set $I$, there are at least $\eps'n$ $|I|$-dimensional hyperplanes from $M_I$, meaning \[
    |P\cap M|\ge \eps'n\cdot (\eps' n)^{|I|}=(\eps' n)^{|I|+1}.
\]
Therefore, $P\in M_{I\cup\{j\}}$, where $j$ is the direction of the line $\ell$, meaning $\ell\subseteq P\subseteq M'$, a contradiction. Thus, $|\ell\cap M_I|<\eps'n$. Then \[
    |\ell\cap M'|\le\sum_{I\subseteq [D]}|\ell\cap M_I|< 2^D\eps'n=\eps n.
\]
Since the line $\ell$ was arbitrary and not completely contained in $M'$, the set $M'$ is $\eps$-closed.
 
 Now, it remains to estimate $|M'|$. It is easy to see that $|M_I|\le (|M|/(\eps'n)^{|I|}) n^{|I|}\le |M|(\eps')^{-|I|}$. Therefore, 
 \begin{align*}
 |[M]_\eps|\le |M'|\le\sum_{I\subseteq [D]}|M|(\eps')^{-|I|}&=|M|(1+1/\eps')^D\\
 &\le |M|\rbr{\frac{2^D+1}{\eps}}^{D},
 \end{align*}
 and we can put $c=\rbr{\frac{2^D+1}{\eps}}^{D}$, completing the proof.
\end{proof}

\begin{lemma}\label{lemma:ig-exp}
    Let $\cC_1,\dots,\cC_D\subsetneq \F_q^n$, $\eps>0$ be such that each $\eps$-closed subset $M\subseteq [n]^D$ for code $\cC:=\cC_1\boxplus\ldots\boxplus\cC_D$ is inner-generated. Then $\rho(\cC_1,\dots,\cC_D)\ge \gamma(\eps,D)$ 
    where $\gamma(\eps,D):=\frac{\eps^D}{D(2^D+1)^D}$.
\end{lemma}
\begin{proof}
    Consider an arbitrary codeword $x\in \cC$ and let $M:=[\supp x]_\eps$. 
    By Lemma \ref{lemma:closure-size} we have $|M|\le c|x|$ where $c=\rbr{\frac{2^D+1}{\eps}}^{D}$. 
    Since $M$ is $\eps$-closed, it is inner-generated; hence, $x=\sum_{\ell\in L(M)}x_\ell$ for some $x_\ell\in\cC_\ell$. 
    Therefore, $x$ is represented as the sum of at most $|L(M)|$ codewords along the lines contained in $M$. 
    Since each line has size $n$, and each point belongs to at most $D$ lines, we have 
    \[|L(M)|\le |M|\frac{D}{n}\le cD\frac{|x|}{n}=\frac{1}{\gamma(\eps,D)}\frac{|x|}{n}.\]
    Therefore, $\rho(\cC_1,\dots,\cC_D)\ge \gamma(\eps,D)$.
\end{proof}
In terms of $\rho$ and $\epsmax$,\[
\rho=\rho(\cC_1,\ldots,\cC_D),\qquad \epsmax=\epsmax(\cC_1,\ldots,\cC_D),
\]
we proved that $\rho$ is bounded from above and below by functions of $\epsmax$:
\[
    \gamma(\epsmax,D)
    \le \rho \le \epsmax,
\]
where the left inequality follows from Lemma~\ref{lemma:ig-exp} and the right inequality follows from Lemma~\ref{lemma:inner-generated}.

The previous lemmas show that good product expansion is equivalent to the condition that all $\eps$-closed sets are inner-generated for some constant $\eps>0$. 
The next lemmas show that for random codes over a~sufficiently large field, $\epsmax$ is bounded from below by a~constant depending only on the rates, with high probability.

\begin{lemma}\label{lemma:universal-exp}
    For every $D\in\NN$ there is a~function $\mu_D\colon (0,1)^D\to (0,1)$ such that for every tuple of rates $(r_1,\dots,r_D)\in(0,1)^D$ and a~maximally extendable code $\cC_1^\perp\otimes\dots\otimes\cC_D^\perp$ such that  $\cC_i\in \Gr_{2^t}(n,k_i)$ and  $k_i\le r_i n$ we have $\rho(\cC_1,\dots,\cC_D)\ge \mu_D(r_1,\dots, r_D)$.
\end{lemma}
\begin{proof}
    Fix $D\in\NN$ and a~tuple of rates $(r_1,\dots,r_D)\in(0,1)^D$. 
    Let $\hat\rho$ be a~constant from Lemma \ref{lemma:prodexp-exists} for rates $r_1,\dots,r_D$. Put $\mu_{D}(r_1,\dots,r_D):=\gamma(\hat\rho,D)$ where $\gamma$ is the function from Lemma \ref{lemma:ig-exp}.

    Consider a~maximally extendable code $\cC_1^\perp\otimes\dots\otimes\cC_D^\perp$ that satisfies the conditions of the lemma. By Lemma \ref{lemma:prodexp-exists} there exist codes $\hat\cC_i\in\Gr_2(n,k_i)$, $i\in[D]$ such that \[\rho(\hat\cC_1,\dots,\hat\cC_D)\ge\hat\rho.\] By Lemma \ref{lemma:inner-generated} each $\hat\rho$-closed subset $M\in[n]^D$ is inner-generated for the code $\hat\cC:=\hat\cC_1\boxplus\dots\boxplus \hat\cC_D$. Since the product code $\cC_1^\perp\otimes\dots\otimes\cC_D^\perp$ is maximally extendable, by Proposition~\ref{pr:1} each $\hat\rho$-closed subset $M\in[n]^D$ is also inner-generated for the code $\cC:=\cC_1\boxplus\dots\boxplus \cC_D$. Therefore, by Lemma \ref{lemma:ig-exp} we have \[\rho(\cC_1,\dots,\cC_D)\ge\gamma(\hat\rho,D)=\mu_D(r_1,\dots,r_D).\] 
\end{proof}

In the rest of this section, we prove that a~randomly chosen product code over a~sufficiently large finite field is maximally extendable with high probability. This is the final key ingredient for our main theorem. The proof closely follows the strategy from~\cite[Lemma~32]{gopalanExplicitMaximallyRecoverable2014}, which involves parameterizing the space of codes with polynomial matrices and then applying the Schwartz-Zippel lemma.

The next lemma follows from Theorem~1 in~\cite{panteleevMaximallyExtendableSheaf2024}, but for the reader's convenience, we provide a~direct proof here, which also comes with better bounds. 
\begin{lemma}\label{lemma:rand-universal}
    For $(\cC_1,\dots,\cC_D)$ picked uniformly at random from $\Gr_{2^t}(n,k_1)\times\cdots\times\Gr_{2^t}(n,k_D)$, the code $\cC_1\otimes\dots\otimes\cC_D$ is maximally extendable with probability at least $1 - n^D 2^{n^D - t + 1}$.
\end{lemma}
\begin{proof}
    The proof proceeds in three main steps. First, we construct a generic parity-check matrix for the product code whose entries are polynomials. Second, we show that substituting variables with random finite field elements yields a~maximally extendable code with high probability. Finally, we connect this result for random matrices back to random codes.

    \textbf{Step 1: Polynomial parameterization.}
    It is clear that the parity-check matrix $H$ for the code $\cC_1\otimes\dots \otimes \cC_D$ can be obtained by stacking the matrices $\otimes_{i\in[D]}M_{i,j}$, $j\in [D]$, where $M_{j,j}=H_j$ and $M_{i,j}=I_{n}$ for $i\ne j$ (each such matrix checks the constraints of the product code in the $j$-th direction). 
    
    Given a~collection of free variables $\vp$, let $\F_2[\vp]$ and $\F_2(\vp)$ denote, respectively, the ring of polynomials and the field of rational functions in the variables from $\vp$.   
    Consider a~parametrization for the codes $\cC_1\otimes\dots \otimes \cC_D$ by \emph{polynomial} parity-check matrices $H_i\in \F_2[\vp]^{(n - k_i)\times n}$, where $(H_i)_{u,v} := p_{i,u,v}$, and $\vp:=(p_{i,u,v})_{i\in [D],u\in [n-k_i],v\in [n]}$ is the collection of $N:=n\sum_{i=1}^D (n-k_i)$ free parameters we used. 

    \textbf{Step 2: Good substitutions and Schwartz-Zippel Lemma.}
    For a~subset of indices $S\subseteq [n]^D$ we consider two submatrices $H_S$ and $H^S$ of $H$. The matrix $H_S$ is obtained from $H$ by retaining only the columns corresponding to the codeword symbols with indices in $S$, while $H^S$ is composed of the rows of $H_S$ corresponding to the parity-check equations in $H$ involving only such symbols.
    We say that a~vector $\va\in\F_{2^t}^N$ is a~\emph{good substitution} for the polynomial matrix $H$ if for each $S\subseteq [n]^D$, we get:
    \begin{equation}\label{eq:good-cond}
    \rk H'(\va)=\rk H'; H'\in\{H_S,H^S\},    
    \end{equation}
    where the left-hand side rank is taken over the finite field $\F_{2^t}$ and the right-hand side over the infinite field $\F_2(\vp)$. Let us denote by $U(2^t)$ the set of all good substitutions over the field $\F_{2^t}$. Note that for $S=\ell\in \cL_i(n,D)$ we have $H^S(\va)=H_i(\va)$, therefore for a good substitution $\va\in U(2^t)$ all matrices $H_i(\va)$ have maximal rank.

    \begin{claim}\label{cl:MR-universal}
        If $\va\in U(2^t)$, then $\cC=\cC_1\otimes\dots \otimes \cC_D$ is maximally extendable, where $\cC_i = \ker H_i(\va)$, $i\in [D]$.
    \end{claim}
    \begin{proof}[Proof of Claim]
    \begin{figure}
        \centering
        \begin{tikzpicture}[scale=0.87]

\draw (0,0) rectangle (10,4);
\fill[gray!30] (0,0) rectangle (3,4);
\node at (1,1.75) {$0$};
\node at (8,1.75) {$0$};
\node at (4,1.75) {$H^S$};
\draw (2,1) rectangle (3,2.5);
\draw[dashed] (0,1) -- (10,1);
\draw[dashed] (0,2.5) -- (10,2.5);
\fill[gray!60] (2,1) rectangle (3,2.5);

\draw (2,1) rectangle (6,2.5);
\draw[dashed] (2,0) -- (2,4);
\draw[dashed] (6,0) -- (6,4);
\draw [decorate,decoration={brace,amplitude=10pt}] (3,0) -- (6,0) node[midway, yshift=0.6cm] {$J$};
\draw [decorate,decoration={brace,mirror,amplitude=12pt}] (3,0) -- (10,0) node[midway, yshift=-0.6cm] {$I$};
\draw [decorate,decoration={brace,amplitude=10pt}] (2,4) -- (6,4) node[midway, yshift=0.6cm] {$S$};
\end{tikzpicture}
        \caption{$S$ is extendable in $\cC=\ker H(\va)$ \Iff $\ker H^S(\va) = \cC|_S$ \Iff any information set $J$ of $\ker H^S(\va)$ is contained in some information set $I$ of $\ker H(\va)$. 
        In this case, the gray submatrices of $H(\va)$ and $H^S(\va)$ both must be full-rank.}
        \label{fig:extendable}
    \end{figure}

    The main idea of the proof is a~simple observation that $S\subseteq [n]^D$ is extendable in $\cC$ \Iff $\cC|_S = \ker H^S(\va)$, where $\cC|_S = \{c|_S \mid c \in \cC\}$ represents the vectors in $\F_{2^t}^S$ that can be extended to full codewords, and the \emph{local code} $\ker H^S(\va)$ represents the vectors in $\F_{2^t}^S$ satisfying all local checks imposed by the codes $\cC_\ell$, $\ell \subseteq S$. Alternatively, $S$ is extendable \Iff every information set $J$ of the local code $\ker H^S(\va)$ can be extended to an~information set $I$ of the global code $\cC = \ker H(\va)$ (see Fig.~\ref{fig:extendable}). 

    Clearly, the inclusion $\cC|_S \subseteq \ker H^S(\va)$ always holds. Thus, in order to show that a~set $S$ is extendable in $\cC$, it is enough to check that \begin{equation}\label{eqn:check-ext}
        \dim \cC|_S \ge \dim \ker H^S(\va).
    \end{equation}
    
    First, note that $\cC = \ker H(\va)$ is \emph{maximally recoverable}~\cite{chenMaximallyRecoverableProperty2007, gopalanExplicitMaximallyRecoverable2014}, which means that for any~other product code $\ker H(\va')$ with the same dimension, if $I$ is an~information set for $\ker H(\va')$, then $I$ is also an~information set for $\cC$. This follows directly from the fact that $\va \in U(2^t)$, and substituting parameters into a~polynomial matrix does not increase its rank:
    \begin{align*}
    n^D - |I| = \rk H_{[n]^D \setminus I}(\va') 
    &\le \rk H_{[n]^D \setminus I}\\ 
    &= \rk H_{[n]^D \setminus I}(\va) \le n^D - |I|.
    \end{align*}
    
    For a~good substitution $\va \in U(2^t)$, the dimension of the local code $\ker H^S(\va)$ is \emph{minimal} among all possible substitutions $\va'$, since $\rk H^S(\va') \le \rk H^S = \rk H^S(\va)$. At the same time, the dimension of the projection $\cC|_S$ is \emph{maximal}, as it equals the largest size of a~set $J \subseteq S$ that can be extended to an~information set of $\cC$. This maximality is guaranteed because $\cC$ is maximally recoverable.
    
    Finally, if $S$ is extendable in another product code $\cC' = \ker H(\va')$ with $\dim \cC' = \dim \cC$, then $S$ is also extendable in $\cC$. This follows because $\cC'|_S = \ker H^S(\va')$, and the maximality of $\dim \cC|_S$ combined with the minimality of $\dim \ker H^S(\va)$ implies 
    \[
    \dim \cC|_S\ge \dim \cC'|_S=\dim \ker  H^S(\va')\ge \dim \ker H^S(\va), 
    \]
    therefore \eqref{eqn:check-ext} holds for $S$ and $\cC|_S = \ker H^S(\va)$.
    Thus, $S$ is extendable in $\cC$, proving that $\cC$ is maximally extendable. 
    \end{proof} 
    \begin{claim}\label{cl:proof}
         \[\PP\fbr{\va \in U(2^t)}\ge 1 - n^D 2^{n^D - t + 1}.\]
    \end{claim}
    \begin{proof}[Proof of Claim]
        Note that we have $2^{n^D + 1}$ polynomial matrices $H_S$ and $H^S$ for all possible $S\subseteq [n]^D$. Let us arbitrarily choose one maximal non-singular submatrix in each such matrix, and denote by $f$ the product of their determinants. 
        From conditions \eqref{eq:good-cond} we see that for $\va\in \F^N_{2^t}$ we have $f(\va)\ne 0$ \Iff $\va\in U(2^t)$. 
        It is also clear that $\deg f \le n^D 2^{n^D + 1}$. 
        Thus, by the Schwartz–Zippel lemma (e.g., see~\cite[Theorem~7.2]{motwaniRandomizedAlgorithms1995}), the probability that for a~uniformly random $\va\in \F^N_{2^t}$ we have $\va \not\in U(2^t)$ is bounded above as \[\deg f/2^t \le n^D 2^{n^D - t + 1},\] and the claim is proved.      
    \end{proof}
 \textbf{Step 3: From Random Matrices to Random Codes.}
 To complete the proof of the lemma we also need to show that the uniform distribution over \emph{codes} is approximately the same as the uniform distribution over \emph{parity-check matrices}. This follows from the following facts:
 \begin{itemize}
    \item If we pick a code $\cC\in\Gr_{2^t}(n,k)$ uniformly at random, and pick a~matrix $H$ for $\cC$ uniformly at random from the set of $(n-k)\times n$ parity-check matrices of the code $\cC$, then the distribution of matrices $H$ obtained this way is uniform over all $(n-k)\times n$ matrices of full rank.
    \item For a good substitution $\va$, all matrices $H_i(\va)$ have maximal rank, thus
    \begin{multline*}
        \PP\rbr{\va\in U(2^t)\mid \dim \ker H_i(\va)=k_i\mbox{ for }i\in[D]}\\
        \ge \PP\rbr{\va\in  U(2^t)}\ge 1 - n^D 2^{n^D - t + 1}.
    \end{multline*}
    \item The parameter vector $\va$ corresponds exactly to the elements of all parity-check matrices $H_i$.
 \end{itemize}
 Thus, if we pick uniformly at random a tuple of codes $(\cC_1,\dots,\cC_D)\in\Gr_{2^t}(n,k_1)\times\cdots\times\Gr_{2^t}(n,k_D)$ and then pick uniformly at random parity-check matrices $H_1,\dots,H_D$ for these codes, with probability at least $1 - n^D 2^{n^D - t + 1}$ the corresponding substitution $\va$ is good. Hence, by Claim~\ref{cl:MR-universal}, the code $\cC_1\otimes\dots \otimes \cC_D$ is maximally extendable.
\end{proof}

Now we are ready to prove the main theorem.
\ThMain*
\begin{proof}
   Put $\rho:=\mu_D(R_1,\dots,R_D)$ where $\mu$ is the function from Lemma \ref{lemma:universal-exp}.
   Fix some numbers $n\in\NN$, $k_i\le R_i n$, $t\in\NN$. 
   By Lemma \ref{lemma:universal-exp}, if for a~collection of codes $(\cC_1,\dots,\cC_D)\in\Gr_{2^t}(n,k_1)\times\cdots\times\Gr_{2^t}(n,k_D)$ the code $\cC_1^\perp\otimes\dots\otimes\cC_D^\perp$ is maximally extendable, then this collection is $\rho$-product-expanding. 
   Moreover, according to Lemma \ref{lemma:rand-universal}, if we pick these codes uniformly at random, the probability of this event is at least $1- n^D 2^{n^D - t + 1}\to 1$ as $t\to\infty$. 
   This completes the proof.
\end{proof}


\bibliographystyle{IEEEtran}
\bibliography{IEEEabrv,codes.bib}

\begin{thebibliography}{10}
\providecommand{\url}[1]{#1}
\csname url@samestyle\endcsname
\providecommand{\newblock}{\relax}
\providecommand{\bibinfo}[2]{#2}
\providecommand{\BIBentrySTDinterwordspacing}{\spaceskip=0pt\relax}
\providecommand{\BIBentryALTinterwordstretchfactor}{4}
\providecommand{\BIBentryALTinterwordspacing}{\spaceskip=\fontdimen2\font plus
\BIBentryALTinterwordstretchfactor\fontdimen3\font minus \fontdimen4\font\relax}
\providecommand{\BIBforeignlanguage}[2]{{%
\expandafter\ifx\csname l@#1\endcsname\relax
\typeout{** WARNING: IEEEtran.bst: No hyphenation pattern has been}%
\typeout{** loaded for the language `#1'. Using the pattern for}%
\typeout{** the default language instead.}%
\else
\language=\csname l@#1\endcsname
\fi
#2}}
\providecommand{\BIBdecl}{\relax}
\BIBdecl

\bibitem{Dinur:stoc2022}
\BIBentryALTinterwordspacing
I.~Dinur, S.~Evra, R.~Livne, A.~Lubotzky, and S.~Mozes, ``Locally testable codes with constant rate, distance, and locality,'' in \emph{Proceedings of the 54th Annual ACM SIGACT Symposium on Theory of Computing}, ser. STOC 2022.\hskip 1em plus 0.5em minus 0.4em\relax New York, NY, USA: Association for Computing Machinery, Jun. 2022, pp. 357--374. [Online]. Available: \url{https://doi.org/10.1145/3519935.3520024}
\BIBentrySTDinterwordspacing

\bibitem{Panteleev&Kalachev:stoc2022}
\BIBentryALTinterwordspacing
P.~Panteleev and G.~Kalachev, ``Asymptotically good quantum and locally testable classical {LDPC} codes,'' in \emph{Proceedings of the 54th Annual ACM SIGACT Symposium on Theory of Computing}, ser. STOC 2022.\hskip 1em plus 0.5em minus 0.4em\relax New York, NY, USA: Association for Computing Machinery, Jun. 2022, pp. 375--388. [Online]. Available: \url{https://doi.org/10.1145/3519935.3520017}
\BIBentrySTDinterwordspacing

\bibitem{Kaufman:2014a}
T.~Kaufman and A.~Lubotzky, ``High dimensional expanders and property testing,'' in \emph{Proceedings of the 5th conference on Innovations in theoretical computer science}, ser. ITCS '14.\hskip 1em plus 0.5em minus 0.4em\relax New York, NY, USA: Association for Computing Machinery, Jan. 2014, pp. 501--506.

\bibitem{Kaufman:2014}
T.~Kaufman, D.~Kazhdan, and A.~Lubotzky, ``Ramanujan complexes and bounded degree topological expanders,'' in \emph{2014 IEEE 55th Annual Symposium on Foundations of Computer Science}, Oct. 2014, pp. 484--493.

\bibitem{Evra:2020}
S.~Evra, T.~Kaufman, and G.~Z{\'e}mor, ``Decodable quantum {{LDPC}} codes beyond the square root distance barrier using high dimensional expanders,'' in \emph{2020 {{IEEE}} 61st {{Annual Symposium}} on {{Foundations}} of {{Computer Science}} ({{FOCS}})}, Nov. 2020, pp. 218--227.

\bibitem{Hastings:2021:fiber}
M.~B. Hastings, J.~Haah, and R.~O'Donnell, ``Fiber bundle codes: breaking the {$N^{1/2} \operatorname{polylog}(N)$} barrier for quantum {LDPC} codes,'' in \emph{Proceedings of the 53rd Annual ACM SIGACT Symposium on Theory of Computing}.\hskip 1em plus 0.5em minus 0.4em\relax New York, NY, USA: Association for Computing Machinery, Jun. 2021, pp. 1276--1288.

\bibitem{Panteleev&Kalachev:2019}
\BIBentryALTinterwordspacing
P.~Panteleev and G.~Kalachev, ``Degenerate quantum {LDPC} codes with good finite length performance,'' \emph{Quantum}, vol.~5, p. 585, Nov. 2021. [Online]. Available: \url{https://quantum-journal.org/papers/q-2021-11-22-585/}
\BIBentrySTDinterwordspacing

\bibitem{Panteleev&Kalachev:2021}
------, ``Quantum {{LDPC Codes with Almost Linear Minimum Distance}},'' \emph{IEEE Transactions on Information Theory}, vol.~68, no.~1, pp. 213--229, Jan. 2022.

\bibitem{Tillich&Zemor:2014}
J.~Tillich and G.~Z{\'e}mor, ``Quantum {LDPC} codes with positive rate and minimum distance proportional to the square root of the blocklength,'' \emph{IEEE Transactions on Information Theory}, vol.~60, no.~2, pp. 1193--1202, Feb. 2014.

\bibitem{qldpc}
D.~J.~C. MacKay, G.~Mitchison, and P.~L. McFadden, ``Sparse-graph codes for quantum error correction,'' \emph{IEEE Transactions on Information Theory}, vol.~50, no.~10, pp. 2315--2330, Oct. 2004.

\bibitem{Hagiwara:2007}
M.~Hagiwara and H.~Imai, ``Quantum quasi-cyclic {LDPC} codes,'' in \emph{2007 IEEE international symposium on information theory}, Jun. 2007, pp. 806--810.

\bibitem{Haah:2011}
J.~Haah, ``Local stabilizer codes in three dimensions without string logical operators,'' \emph{Physical Review A}, vol.~83, no.~4, p. 042330, Apr. 2011.

\bibitem{Kovalev:2013}
A.~A. Kovalev and L.~P. Pryadko, ``Quantum kronecker sum-product low-density parity-check codes with finite rate,'' \emph{Physical Review A}, vol.~88, no.~1, p. 012311, Jul. 2013.

\bibitem{Zemor:2001}
G.~{Z\'{e}mor}, ``On expander codes,'' \emph{IEEE Transactions on Information Theory}, vol.~47, no.~2, pp. 835--837, 2001.

\bibitem{Sipser:1996}
M.~Sipser and D.~Spielman, ``Expander codes,'' \emph{IEEE Transactions on Information Theory}, vol.~42, no.~6, pp. 1710--1722, Nov. 1996.

\bibitem{Spielman2009SpectralGraphTheory}
\BIBentryALTinterwordspacing
D.~Spielman, ``Spectral graph theory, fall 2009,'' Course Notes, Applied Mathematics 561/Computer Science 662, Yale University, 2009. [Online]. Available: \url{https://www.cs.yale.edu/homes/spielman/561/2009/lect12-09.pdf}
\BIBentrySTDinterwordspacing

\bibitem{dworkCS369EExpandersComputer}
P.~Harsha, ``{CS369E}: {Expanders} in {Computer Science},'' \url{https://www.tcs.tifr.res.in/~prahladh/teaching/05spring/}, 2005, accessed: Oct. 19, 2025.

\bibitem{Tanner:1981}
R.~Tanner, ``A recursive approach to low complexity codes,'' \emph{IEEE Transactions on Information Theory}, vol.~27, no.~5, pp. 533--547, 1981.

\bibitem{Margulis:1988}
G.~A. Margulis, ``Explicit group-theoretical constructions of combinatorial schemes and their application to the design of expanders and concentrators,'' \emph{Problemy peredachi informatsii}, vol.~24, no.~1, pp. 51--60, 1988.

\bibitem{Lubotzky:1988}
\BIBentryALTinterwordspacing
A.~Lubotzky, R.~Phillips, and P.~Sarnak, ``Ramanujan graphs,'' \emph{Combinatorica}, vol.~8, no.~3, pp. 261--277, Sep. 1988. [Online]. Available: \url{https://doi.org/10.1007/BF02126799}
\BIBentrySTDinterwordspacing

\bibitem{kalachevTwosidedRobustlyTestable2023}
\BIBentryALTinterwordspacing
G.~Kalachev and P.~Panteleev, ``Two-sided robustly testable codes,'' Jun. 2022, unpublished. [Online]. Available: \url{http://arxiv.org/abs/2206.09973}
\BIBentrySTDinterwordspacing

\bibitem{LiftingSmallLocally2019}
P.~Harsha, ``Lifting small locally testable codes ({{LTCs}}) to large {{LTCs}} via {{HDXs}}, {{Institute}} for {{Advanced Study}},'' \href{https://www.ias.edu/video/csdm/2019/1125-PrahladhHarsha}{Video}, Nov. 2019.

\bibitem{Leverrier:focs2022}
\BIBentryALTinterwordspacing
A.~Leverrier and G.~Z{\'e}mor, ``Quantum {T}anner codes,'' in \emph{2022 IEEE 63rd Annual Symposium on Foundations of Computer Science (FOCS)}.\hskip 1em plus 0.5em minus 0.4em\relax Los Alamitos, CA, USA: IEEE Computer Society, Nov. 2022, pp. 872--883. [Online]. Available: \url{https://doi.ieeecomputersociety.org/10.1109/FOCS54457.2022.00117}
\BIBentrySTDinterwordspacing

\bibitem{Dinur:decoders}
\BIBentryALTinterwordspacing
I.~Dinur, M.-H. Hsieh, T.-C. Lin, and T.~Vidick, ``Good quantum {LDPC} codes with linear time decoders,'' in \emph{Proceedings of the 55th Annual ACM Symposium on Theory of Computing}, ser. STOC 2023.\hskip 1em plus 0.5em minus 0.4em\relax New York, NY, USA: Association for Computing Machinery, 2023, p. 905–918. [Online]. Available: \url{https://doi.org/10.1145/3564246.3585101}
\BIBentrySTDinterwordspacing

\bibitem{Gu:stoc2023:qpdpc-decoder}
\BIBentryALTinterwordspacing
S.~Gu, C.~A. Pattison, and E.~Tang, ``An efficient decoder for a linear distance quantum {LDPC} code,'' in \emph{Proceedings of the 55th Annual ACM Symposium on Theory of Computing}, ser. STOC 2023.\hskip 1em plus 0.5em minus 0.4em\relax New York, NY, USA: Association for Computing Machinery, 2023, p. 919–932. [Online]. Available: \url{https://doi.org/10.1145/3564246.3585169}
\BIBentrySTDinterwordspacing

\bibitem{Leverrier:qldpcdecoder:2023}
\BIBentryALTinterwordspacing
A.~Leverrier and G.~Z{\'e}mor, ``Efficient decoding up to a constant fraction of the code length for asymptotically good quantum codes,'' in \emph{Proceedings of the 2023 Annual ACM--SIAM Symposium on Discrete Algorithms (SODA)}, N.~Bansal and V.~Nagarajan, Eds.\hskip 1em plus 0.5em minus 0.4em\relax Philadelphia, PA, USA: Society for Industrial and Applied Mathematics (SIAM), 2023, pp. 1216--1244. [Online]. Available: \url{https://epubs.siam.org/doi/abs/10.1137/1.9781611977554.ch45}
\BIBentrySTDinterwordspacing

\bibitem{guSingleShotDecodingGood2024}
S.~Gu, E.~Tang, L.~Caha, S.~H. Choe, Z.~He, and A.~Kubica, ``Single-{{Shot Decoding}} of {{Good Quantum LDPC Codes}},'' \emph{Communications in Mathematical Physics}, vol. 405, no.~3, p.~85, Mar. 2024.

\bibitem{Aharonov:2015}
D.~Aharonov and L.~Eldar, ``Quantum locally testable codes,'' \emph{SIAM Journal on Computing}, vol.~44, no.~5, pp. 1230--1262, Jan. 2015.

\bibitem{eldarLocalHamiltoniansWhose2017}
L.~Eldar and A.~W. Harrow, ``Local {{Hamiltonians Whose Ground States Are Hard}} to {{Approximate}},'' in \emph{2017 {{IEEE}} 58th {{Annual Symposium}} on {{Foundations}} of {{Computer Science}} ({{FOCS}})}.\hskip 1em plus 0.5em minus 0.4em\relax {Berkeley, CA, USA}: {IEEE}, Oct. 2017, pp. 427--438.

\bibitem{crossQuantumLocallyTestable2024}
A.~Cross, Z.~He, A.~Natarajan, M.~Szegedy, and G.~Zhu, ``Quantum {{Locally Testable Code}} with {{Constant Soundness}},'' \emph{Quantum}, vol.~8, p. 1501, Oct. 2024.

\bibitem{dinurqLTC:2024}
I.~Dinur, T.-C. Lin, and T.~Vidick, ``Expansion of {{High-Dimensional Cubical Complexes}}: With {{Application}} to {{Quantum Locally Testable Codes}},'' in \emph{2024 {{IEEE}} 65th {{Annual Symposium}} on {{Foundations}} of {{Computer Science}} ({{FOCS}})}, Oct. 2024, pp. 379--385.

\bibitem{golowichQuantumLDPCCodes2024}
L.~Golowich and V.~Guruswami, ``Quantum {{LDPC Codes}} of {{Almost Linear Distance}} via {{Iterated Homological Products}},'' in \emph{40th {{Computational Complexity Conference}} ({{CCC}} 2025)}, ser. Leibniz {{International Proceedings}} in {{Informatics}} ({{LIPIcs}}), S.~Srinivasan, Ed., vol. 339.\hskip 1em plus 0.5em minus 0.4em\relax Dagstuhl, Germany: Schloss Dagstuhl -- Leibniz-Zentrum f{\"u}r Informatik, 2025, pp. 25:1--25:11.

\bibitem{nguyenQuantumFaultTolerance2024}
Q.~T. Nguyen and C.~A. Pattison, ``Quantum {{Fault Tolerance}} with {{Constant-Space}} and {{Logarithmic-Time Overheads}},'' in \emph{Proceedings of the 57th {{Annual ACM Symposium}} on {{Theory}} of {{Computing}}}, ser. {{STOC}} '25.\hskip 1em plus 0.5em minus 0.4em\relax New York, NY, USA: Association for Computing Machinery, Jun. 2025, pp. 730--737.

\bibitem{Kalachev:example:2023}
\BIBentryALTinterwordspacing
G.~Kalachev, ``High-dimensional expansion of product codes is stronger than robust and agreement testability,'' Aug. 2023, unpublished. [Online]. Available: \url{https://arxiv.org/abs/2308.02889}
\BIBentrySTDinterwordspacing

\bibitem{panteleevMaximallyExtendableSheaf2024}
\BIBentryALTinterwordspacing
P.~Panteleev and G.~Kalachev, ``Maximally {{Extendable Sheaf Codes}},'' Mar. 2024, unpublished. [Online]. Available: \url{https://doi.org/10.48550/arXiv.2403.03651}
\BIBentrySTDinterwordspacing

\bibitem{Wolf:1965}
J.~Wolf, ``On codes derivable from the tensor product of check matrices,'' \emph{IEEE Transactions on Information Theory}, vol.~11, no.~2, pp. 281--284, Apr. 1965.

\bibitem{Chien:1973}
R.~Chien and S.~Ng, ``Dual product codes for correction of multiple low-density burst errors,'' \emph{IEEE Transactions on Information Theory}, vol.~19, no.~5, pp. 672--677, Sep. 1973.

\bibitem{Linial:2006}
\BIBentryALTinterwordspacing
N.~Linial and R.~Meshulam, ``Homological connectivity of random 2-complexes,'' \emph{Combinatorica}, vol.~26, no.~4, pp. 475--487, Aug. 2006. [Online]. Available: \url{https://doi.org/10.1007/s00493-006-0027-9}
\BIBentrySTDinterwordspacing

\bibitem{Gromov:2010}
\BIBentryALTinterwordspacing
M.~Gromov, ``Singularities, expanders and topology of maps. part 2: from combinatorics to topology via algebraic isoperimetry,'' \emph{Geometric and Functional Analysis}, vol.~20, no.~2, pp. 416--526, Aug. 2010. [Online]. Available: \url{https://doi.org/10.1007/s00039-010-0073-8}
\BIBentrySTDinterwordspacing

\bibitem{dotterrerCoboundaryExpanders2012}
D.~Dotterrer and M.~Kahle, ``Coboundary expanders,'' \emph{Journal of Topology and Analysis}, vol.~04, no.~04, pp. 499--514, Dec. 2012.

\bibitem{Goldreich:2010}
O.~Goldreich, ``\BIBforeignlanguage{en}{Short locally testable codes and proofs: A survey in two parts},'' in \emph{\BIBforeignlanguage{en}{Property Testing: Current Research and Surveys}}, ser. Lecture Notes in Computer Science, O.~Goldreich, Ed.\hskip 1em plus 0.5em minus 0.4em\relax Berlin, Heidelberg: Springer, 2010, pp. 65--104.

\bibitem{Leverrier:2021a}
A.~Leverrier, V.~Londe, and G.~Z{\'e}mor, ``Towards local testability for quantum coding,'' \emph{Quantum}, vol.~6, p. 661, Feb. 2022.

\bibitem{chenMaximallyRecoverableProperty2007}
M.~Chen, C.~Huang, and J.~Li, ``On the {{Maximally Recoverable Property}} for {{Multi-Protection Group Codes}},'' in \emph{2007 {{IEEE International Symposium}} on {{Information Theory}}}, Jun. 2007, pp. 486--490.

\bibitem{gopalanExplicitMaximallyRecoverable2014}
P.~Gopalan, C.~Huang, B.~Jenkins, and S.~Yekhanin, ``Explicit {{Maximally Recoverable Codes With Locality}},'' \emph{IEEE Transactions on Information Theory}, vol.~60, no.~9, pp. 5245--5256, Sep. 2014.

\bibitem{Ben-Sasson:2006}
\BIBentryALTinterwordspacing
E.~{Ben-Sasson} and M.~Sudan, ``Robust locally testable codes and products of codes,'' \emph{Random Structures \& Algorithms}, vol.~28, no.~4, pp. 387--402, 2006. [Online]. Available: \url{https://onlinelibrary.wiley.com/doi/abs/10.1002/rsa.20120}
\BIBentrySTDinterwordspacing

\bibitem{Lin2022:losslessLTC}
T.-C. Lin and M.-H. Hsieh, ``c3-locally testable codes from lossless expanders,'' in \emph{2022 IEEE International Symposium on Information Theory (ISIT)}, 2022, pp. 1175--1180.

\bibitem{motwaniRandomizedAlgorithms1995}
R.~Motwani and P.~Raghavan, \emph{Randomized {{Algorithms}}}.\hskip 1em plus 0.5em minus 0.4em\relax {Cambridge University Press}, Aug. 1995.

\end{thebibliography}

\appendix

\section{LTCs of Arbitrary Length}\label{ap:LTC}

In this section, we give the proof of Theorem~\ref{th:LTC} by constructing good LTCs of arbitrary length out of the good LTCs proposed by Dinur~et al.~\cite{Dinur:stoc2022}. 
However, first, we give an~extended formulation of Theorem~1.1 from~\cite{Dinur:stoc2022}, where we also include additional information about the constructed LTCs that appeared in its proof.
\begin{exttheorem}\label{th:Dinur:LTC}
    For every $r\in (0,1)$ there exist constants $s>0$, $\delta>0$, $\Delta\in\NN$, and $q\in\NN$ such that for all $i\in \NN$ there exist $(\Delta,s)$-locally testable $[n_i,k_i,d_i]$ codes over $\F_2$ with $\Delta$-limited parity-check matrix $H_i$ for some $k_i\ge rn_i$ and $d_i\ge\delta n_i$, where $n_i:=\frac18 \Delta^2 (q^{3i}-q^i)$.
\end{exttheorem}

\begin{remark}\label{rm:Dinur:LTC}
In Theorem 1.1 of~\cite{Dinur:stoc2022}, LTCs are constructed using the square complex $\cay(A,G,B)$ with two generating sets $A$ and $B$ of size $\Delta$, where the bits are assigned to squares and the checks to the edges. The row weights in the parity-check matrix $H_i$ of the resulting $[n_i,k_i,d_i]$ code do not exceed $\Delta$. Each local code has rate no less than $3/4$, so its parity-check matrix has no more than $\Delta/4$ rows, and each square participates in $4$ local codes corresponding to the edges of the square. Therefore, each square is involved in at most $\Delta$ checks; thus, the matrix $H_i$ is $\Delta$-limited. It will be convenient to choose an overcomplete parity-check matrix for local codes by duplicating some of its rows multiple times so that it has exactly $\floorbr{\Delta/4}$ rows; in this case, the matrix $H_i$ of the entire LTC will have $n$ columns and $m=\frac{4\floorbr{\Delta/4}}{\Delta}n$ rows. Assuming that the rate is $<1$, we get $\Delta\geq 4$, $n/2<m\leq n$.    
\end{remark}

\begin{remark}
It is also worth noting that the test in~\cite{Dinur:stoc2022}, which checks the entire set of rows in the parity-check matrix corresponding to the product of local codes, is slightly different from our definition of a~local code, where exactly one row of the parity-check matrix is checked. Therefore, our parameters may differ by a~constant factor from the parameters of the test from~\cite{Dinur:stoc2022}. In particular, according to our definition, the locality of the constructed code is $\Delta$, while in the test from~\cite{Dinur:stoc2022} the locality is $\Delta^2$. Furthermore, our soundness parameter $s$ differs from the parameter $\kappa$ in~\cite{Dinur:stoc2022}. 
\end{remark}

We will also need a~lemma that the direct sum of locally testable codes is also a~locally testable code.

\begin{lemma}\label{lemma:sum-LTC}
    Let $\cC \subseteq \F_2^{n}$ be a~$(\Delta,s)$-locally testable code of length $n$ with a~$\Delta$-limited parity-check matrix of size $m \times n$. Then for any $t\in\NN$ the code $\cC\otimes \F_2^t$ is also a~$(\Delta,s)$-locally testable code with a~$\Delta$-limited parity-check matrix of size $tm \times tn$.
\end{lemma}
\begin{proof}
    Let $H$ be the parity-check matrix of the code $\cC$. Consider the $m'\times n'$ parity-check matrix of the code $\cC':=\cC\otimes \F_2^t$: 
    \[H':=H\otimes I_t=\begin{pmatrix}
        H & 0 & \cdots & 0\\ 
        0 & H & \cdots & 0\\
        \vdots&\vdots&\ddots&\vdots\\
        0&0&\cdots&H
    \end{pmatrix},\]
    where $I_t$ is the identity matrix of size $t\times t$, $m'=tm$, $n'=tn$.

    Consider an arbitrary $x=(x_1,\ldots,x_t)\in \F_2^{n'}$, $x_i\in \F_2^{n}$ for $i\in[t]$. 
    Then we have
    \begin{align*}
        d(x,\cC')&=\sum_{i=1}^t d(x_i,\cC)
        \le \sum_{i=1}^t\frac{1}{s}\cdot\frac{n}{m}|H x_i|\\
        &= \frac{1}{s}\cdot\frac{n'}{m'}\sum_{i=1}^t|H x_i|
        = \frac{1}{s}\cdot\frac{n'}{m'}|H' x|.
    \end{align*}
    Furthermore, it is easy to see that the maximum weight of the rows and columns in the matrix $H'$ is the same as in the matrix $H$.   
    The lemma is proved.
\end{proof}

Additionally, we will need a~lemma that the local testability is preserved when adding a~zero code.

\begin{lemma}\label{lemma:pad-LTC}
    Consider a~$(\Delta,s)$-locally testable code $\cC \subsetneq \F_2^{n}$ with a~$\Delta$-limited parity-check matrix $H$ of size $m \times n$, $n/2 \le m \le n$, and let $\mathbf{0}_t \subseteq \F_2^t$ be the zero code of length $t$. Then the code $\cC' := \cC \oplus \mathbf{0}_t$ is a~$(\Delta, s')$-locally testable code with a~$\Delta$-limited parity-check matrix where $s':=\min(\frac{s}{2}, 1)$.
\end{lemma}
\begin{proof}
    Consider the parity-check matrix of the code $\cC'$:
    \[
    H' := \begin{pmatrix}
        H & 0 \\ 0 & I_t
    \end{pmatrix},
    \]
    where $I_t$ is the identity matrix of size $t \times t$. 
    Let $n' := n + t$, $m' := m + t$ be the number of columns and the number of rows in the matrix $H'$, respectively.
    From the constraints on $m, n$, it is easy to see that $1 \le \frac{n'}{m'} \le \frac{n}{m} \le 2 \frac{n'}{m'}$.
   
    Consider an arbitrary vector $w = (x, y) \in \F_2^{n'}$ where ${x \in \F_2^{n}}$, $y \in \F_2^t$. 
    Then we have:
    \begin{align*}
        d(w, \cC') &= d(x, \cC) + |y| 
        \le \frac{1}{s} \cdot \frac{n}{m} |H x| + |y| \\
        &\le \frac{1}{s} \cdot 2 \frac{n'}{m'} |H x| + \frac{n'}{m'} |I_t y| \\
        &\le \max \left( \frac{2}{s}, 1 \right) \frac{n'}{m'} (|H x| + |y|) 
        = \frac{1}{s'} \cdot \frac{n'}{m'} |H' w|.
    \end{align*}
    It is easy to see that the weight of the columns and rows of the matrix $H'$ is bounded by $\max(\Delta, 1) = \Delta$.
    The lemma is proved.
\end{proof}

\begin{proof}[Proof of Theorem \ref{th:LTC}]
    Let $r := (1 + R) / 2$. According to Theorem \ref{th:Dinur:LTC}, there exist constants $s_0 > 0$, $\delta_0 > 0$, ${\Delta \in \NN}$, ${q \in \NN}$ such that for any $i \in \NN$ there exists a~$(\Delta, s_0)$-locally testable $[n_i, k_i, d_i]$ code $\cC_i$ with a~$\Delta$-limited parity-check matrix $H_i$ of size $m_i \times n_i$ (according to Remark~\ref{rm:Dinur:LTC}, we can ensure that $n_i / 2 < m_i \le n_i$), where $n_i = \frac{1}{8} \Delta^2 (q^{3i} - q^i)$, $k_i \ge r n_i$, $d_i \ge \delta n_i$.

    Let us define 
    \begin{align*}
        &\delta := \min \left( \frac{1 - r}{n_1}, \delta_0 \frac{1-r}{q^3 + q} \right), && s := \min \left( \frac{1 - r}{n_1}, \frac{s_0}{2} \right),
    \end{align*}
    and show that for any $n \in \NN$ there exists a~code with the required parameters.
    
    Fix an arbitrary $n \in \NN$.
    If $n < n_1 / (1 - r)$, then the code $\cC = \F_2^n$ with rate $1$ has distance $1 > \delta n$. This trivial code is $(\Delta, s)$-locally testable since the distance from any vector to the code is zero. It can be defined by a zero $n\times n$ parity-check matrix to satisfy the statement of the theorem.
    
    Next, consider the main case $n \ge n_1 / (1 - r)$. 
    Note that for all $i \in \NN$ the following holds
    \[
    n_{i+1} / n_i = q \frac{q^{2i+2} - 1}{q^{2i} - 1} \leq q (q^2 + 1).
    \]
    Let us put
    \[
    j := \max \left\{ i \in \NN \mid n_i \le (1 - r) n \right\}.
    \] 
    Then $j \ge 1$ and $n_j \le n (1 - r) < n_{j+1}$, which implies \[
        n_j \ge \frac{n_{j+1}}{q^3 + q} > \frac{1 - r}{q^3 + q} n.
    \] 
    Let $n = n_j t + u$, $t \in \NN$, $0 \le u < n_j$. 
    Consider the codes $\cC' := \cC_j\otimes \F_2^t$ and $\cC := \cC' \oplus \mathbf{0}_u$. 
    Let us find the parameters of the code $\cC$.
    It is easy to see that the code $\cC'$ has length $tn_j = n - u$, and the code $\cC$ has length $n$.
    For the dimension $k := \dim \cC$, we get
    \begin{align*}
    k &= \dim \cC' = t \dim \cC_j = t k_j 
    \ge t n_j r = n r - u r\\
    & > n r - n (1 - r) r = nr^2 > n (2r - 1) = n R.
    \end{align*}
    For the distance we have 
    \[
    d(\cC) = d(\cC') = d(\cC_j) \ge \delta_0 n_j > \delta_0 \frac{1 - r}{q^3 + q} n = \delta n.
    \]

    By Lemma \ref{lemma:sum-LTC}, code $\cC'$ is a $(\Delta, s_0)$-locally testable code with a $\Delta$-limited parity-check matrix $H'$ of size $m' \times n'$, $n' / m' = n_j / m_j \in [1, 2)$. Therefore, by Lemma \ref{lemma:pad-LTC}, the code $\cC$ is a~$(\Delta, \min(1, \frac{s_0}{2}))$-locally testable code with a $\Delta$-limited parity-check matrix.
    Thus, we have constructed a~code with the required parameters, and the theorem is proved.
\end{proof}

\end{document}